\documentclass[11pt]{article}

\usepackage{amsmath}
\usepackage{amssymb}
\usepackage{amsthm}
\usepackage{graphics}
\usepackage[latin1]{inputenc}
\usepackage[english]{babel}
\usepackage[T1]{fontenc}

\usepackage{enumerate}

\newtheorem{thm}{Theorem}[section]
\newtheorem{prop}{Proposition}[section]
\newtheorem{coro}{Corollary}[section]
\newtheorem{lemma}{Lemma}[section]
\newtheorem{rem}{Remark}[section]
\newtheorem{rems}{Remarks}[section]
\newtheorem{defi}{Definition}[section]

\newcommand{\R}{\mathbb{R}}             % REAL
\newcommand{\N}{\mathbb{N}}             % INTEGER
\newcommand{\Z}{\mathbb{Z}}             % ???
\newcommand{\C}{\mathbb{C}}             % COMPLEX
\renewcommand{\H}{\mathcal{H}}          % GENERIC HILBERT SPACE

\newcommand{\B}{\mathcal{B}}            % BOUNDED OPERATOR
\newcommand{\M}{\mathcal{M}}            % MANIFOLD
\renewcommand{\S}{\mathcal{S}}          % SURFACE

\newcommand{\e}{\epsilon}
\newcommand{\eo}{\epsilon_{0}}
\newcommand{\eu}{\epsilon_{1}}
\newcommand{\half}{\frac{1}{2}}

\newcommand{\kl}{\langle}
\newcommand{\er}{\rangle}

\newcommand{\A}{\mathcal{A}}               % Angular Dirac Operator

\newcommand{\ds}{\displaystyle}

\newcommand{\Section}[1]{\section{#1} \setcounter{equation}{0}}

\begin{document}

\title{Inverse scattering at fixed energy on asymptotically hyperbolic Liouville surfaces}
\author{Thierry Daud\'e \footnote{Research supported by the French National Research Projects AARG, No. ANR-12-BS01-012-01, and Iproblems, No. ANR-13-JS01-0006 and by the UMI centre de recherches math\'ematiques 3457} $^{\,1}$, Niky Kamran \footnote{Research supported by NSERC grant RGPIN 105490-2011} $^{\,2}$ and Francois Nicoleau \footnote{Research supported by the French National Research Project NOSEVOL, No. ANR- 2011 BS0101901} $^{\,3}$\\[12pt]
 $^1$  \small D\'epartement de Math\'ematiques. UMR CNRS 8088. \\
\small Universit\'e de Cergy-Pontoise \\
\small 95302 Cergy-Pontoise, France  \\
\small thierry.daude@u-cergy.fr\\
$^2$ \small Department of Mathematics and Statistics\\
\small  McGill
University\\ \small Montreal, QC, H3A 2K6, Canada\\
\small nkamran@math.mcgill.ca\\
$^3$  \small  Laboratoire de Math\'ematiques Jean Leray, UMR CNRS 6629 \\
      \small  2 Rue de la Houssini\`ere BP 92208 \\
      \small  F-44322 Nantes Cedex 03 \\
\small francois.nicoleau@math.univ-nantes.fr \\}

%%% TITLE %%%

%%% DATE %%%

%\date{\today}

%%% TITLE %%%

\maketitle

%%% ABSTRACT %%%

\begin{abstract}

In this paper, we study an inverse scattering problem on Liouville surfaces having two asymptotically hyperbolic ends. The main property of Liouville surfaces consists in the complete separability of the Hamilton-Jacobi equations for the geodesic flow. An important related consequence is the fact that the stationary wave equation can be separated into a system of a radial and angular ODEs. The full scattering matrix at fixed energy associated to a scalar wave equation on asymptotically hyperbolic Liouville surfaces can be thus simplified by considering its restrictions onto the generalized harmonics corresponding to the angular separated ODE. The resulting partial scattering matrices consists in a countable set of $2 \times 2$ matrices whose coefficients are the so called transmission and reflection coefficients. It is shown that the reflection coefficients are nothing but generalized Weyl-Titchmarsh functions for the radial ODE in which the generalized angular momentum is seen as the spectral parameter. Using the Complex Angular Momentum method and recent results on 1D inverse problem from generalized Weyl-Titchmarsh functions, we show that the knowledge of the reflection operators at a fixed non zero energy is enough to determine uniquely the metric of the asymptotically hyperbolic Liouville surface under consideration.

%%% KEYWORDS %%%

\vspace{2cm}

\noindent \textit{Keywords}. Inverse scattering, Liouville surfaces, wave equation, singular Sturm-Liouville problems, Weyl-Titchmarsh function.

%%% SUBJECTS CLASSIFICATION %%%%

\noindent \textit{2010 Mathematics Subject Classification}. Primaries 81U40, 35P25; Secondary 58J50.

\end{abstract}

\tableofcontents
\newpage

%%%%%%%%%%%%%%%%%%%%%%%%%%%%%%%%%%%%%%%%%%%%% INTRODUCTION %%%%%%%%%%%%%%%%%%%%%%%%%%%%%%%%%%%%%%%%%%%%%%%%%%%%%%%%%%%%%%%%%%%%%

\Section{Introduction. Statement of the main results.} \label{Intro}

A Liouville surface is a surface equipped with a Riemannian metric which, in a local coordinate system $(x,y)$, has the form
\begin{equation} \label{LiouvilleMetric}
  g = \left(a(x) - b(y)\right) \left[ dx^2 + dy^2 \right],
\end{equation}
where the functions $a(x)$ and $b(y)$ are smooth scalar functions depending only on \emph{one} of the local coordinates $x$ and $y$ respectively. A characteristic property of Liouville surfaces is that the corresponding geodesic flow is completely integrable, a remarkable fact since the isometry group of a Liouville surface will in general be trivial. The complete integrability of the geodesic flow on Liouville surfaces is due to the presence of hidden symmetry (in the sense of the existence of an additional first integral) which originates from the separability of the Hamilton-Jacobi equation for the geodesic flow. More precisely, denoting by $q = (x,y)$ the space variables and by $p = (p_x, p_y)$ the momenta, the Hamiltonian system for the geodesic flow reads
\begin{eqnarray*}
    \dot{p_x}  & = & \frac{H(p,q) \, a'(x)}{a(x) - b(y)}\ , \  \dot{p_y}  = -\frac{H(p,q) \, b'(y)}{a(x) - b(y)}, \\
	\dot{x}    & = & \frac{2p_x}{a(x) - b(y)}\ ,   \ \ \ \            \dot{y}  = \frac{2p_y}{a(x) - b(y)},
\end{eqnarray*}
where the Hamiltonian function is given by
$$
  H(p,q) = \frac{p_x^2 + p_y^2}{a(x) - b(y)}.
$$
The first integral takes then the form
$$
  S(p,q) = \frac{b(y)}{a(x) - b(y)} p_x^2 + \frac{a(x)}{a(x) - b(y)} p_y^2.
$$

Another (closely related) property of Liouville surfaces that will be constantly used in this paper is the fact that the (stationary) wave equation
\begin{equation} \label{WaveEq}
  - \triangle_g f = \lambda^2 f, \quad \lambda \in \R,
\end{equation}
where $\triangle_g$ denotes the Laplace-Beltrami operator on $(S,g)$, can be separated into a system of two ODEs in the variables $x$ and $y$ respectively. Assuming that $f(x,y) = u(x) v(y)$, the stationary wave equation (\ref{WaveEq}) is readily shown to be equivalent to
\begin{equation} \label{Separated-ODE}
 \left\{ \begin{array}{ccl} - \frac{d^2 u(x)}{dx^2} - \lambda^2 a(x) u(x) & = &-\mu^2 u(x), \\
	- \frac{d^2 v(y)}{dy^2} + \lambda^2 b(y) v(y) & = & \mu^2 v(y),
	\end{array} \right.
\end{equation}
where $\mu^2 \in \C$ denotes the constant of separation. Underlying the separation of variables is the observation that the second order differential operator
$$
  S = -\frac{b(y)}{a(x) - b(y)} \partial_x^2 - \frac{a(x)}{a(x) - b(y)} \partial_y^2,
$$
obtained from the first integral in the usual way replacing the momenta $p_x, p_y$ by the differential operators $-i\partial_x, -i\partial_y$, commutes with the Laplace-Beltrami operator (see for instance \cite{KS}), \textit{i.e.}
$$
  [\triangle_g, S] = 0.
$$
We can thus try to find functions which are simultaneous eigenfunctions of both the operators $-\Delta_{g}$ and $S$, \textit{i.e.} we look for functions $f$ that satisfy
$$
\left\{ \begin{array}{ccc}
  -\frac{1}{a(x) - b(y)} \partial_x^2 f - \frac{1}{a(x) - b(y)} \partial_y^2 f & = & \lambda^2 f, \\
	-\frac{b(y)}{a(x) - b(y)} \partial_x^2 f - \frac{a(x)}{a(x) - b(y)} \partial_y^2 f & = & \mu^2 f.
\end{array} \right.
$$
This system is equivalent to the system
$$
\left\{ \begin{array}{ccc}
  - \partial_x^2 f - ( \lambda^2 a(x) - \mu^2) f & = & 0,  \\
	- \partial_y^2 f + ( \lambda^2 b(y) - \mu^2) f & = & 0,
\end{array} \right.
$$
which shows in turn that the eigenfunctions $f$ have the form $f(x,y) = u(x) v(y)$ where $u, v$ satisfy the above system of ODEs (\ref{Separated-ODE}).

We emphasize that any Riemannian surface enjoying such a \emph{separation of variables} property can be written locally under the form of a Liouville surface as above. We refer to Section \ref{LiouvilleSurfaces} for a more complete discussion on this point.

In this paper, we shall investigate the scattering and inverse scattering problem at fixed energy for Liouville surfaces having two asymptotically hyperbolic ends. Precisely and following the models described in Isozaki and Kurylev \cite{IK}, we shall assume that our Liouville surfaces satisfy the following definition.

\begin{defi}[Asymptotically hyperbolic Liouville surfaces] \label{AHLS}
A Liouville surface with two asymptotically hyperbolic ends is a surface $\S$ covered by a global chart
\begin{equation}\label{LiouvilleSurface}
  \S = (0,A)_x \times (0,B)_y,
\end{equation}
where $x \in (0,A)$ corresponds to a boundary defining function for the two asymptotically hyperbolic ends $\{x = 0\}$ and $\{x = A\}$ and $y \in (0,B)$ is as an angular variable on $S^1_B$, the circle of radius $B$. The surface $\S$ is equipped with a Liouville metric $g$ of the form (\ref{LiouvilleMetric}) in which the functions $a(x)$ and $b(y)$ are assumed to be smooth functions satisfying \\
\begin{enumerate}[(i)]

\item $\quad \forall x,y \in (0,A)_x \times (0,B)_y, \quad a(x) - b(y) > 0, \ \ (g \ \textrm{metric})$.

\item $\quad b^{(i)}(B) = b^{(i)}(0), \ \forall i \in \N, \ \ (\textrm{periodic conditions on} \ b \ \textrm{since} \ y \in S^1_B)$.

\item \ \textrm{Asymptotically hyperbolic ends at $\{x=0\}$ and $\{x=A\}$:}
\par
$\exists \ \e_0, \e_1, \delta > 0, \ \textrm{such that} \ \ \forall \alpha, n \in \N, \ \ \exists C_{\alpha n} > 0 \ \textrm{such that} $	
\begin{enumerate}
\item $\quad \forall x \in (0, A - \delta), \ \forall y \in (0,B)$,
\begin{equation} \label{AH-0}		
	\left| \partial_y^\alpha \left(x \partial_x\right)^n \left[ x^2 (a(x) - b(y)) - 1 \right] \right| \leq C_{\alpha n} \left( 1 + |\log x| \right)^{-1 - \e_0 - n},
\end{equation}
and \\

\item $\quad \forall x \in (\delta, A), \ \forall y \in (0,B)$,
\begin{equation} \label{AH-A}		
	\hspace{-0.3cm} \left| \partial_y^\alpha \left((A-x) \partial_x\right)^n \left[ (A-x)^2 (a(x) - b(y)) - 1 \right] \right| \leq C_{\alpha n} \left( 1 + |\log (A-x)| \right)^{-1 - \e_1 - n}.
\end{equation}
\end{enumerate}
\end{enumerate}
\end{defi}

We now explain the meaning of asymptotically hyperbolic ends for our Liouville surfaces. Observe that the metric (\ref{LiouvilleMetric}) can be written in a neighbourhood of $\{x=0\}$ in the form
$$
  g = \frac{dx^2 + dy^2 + P(x,y,dx,dy)}{x^2},
$$
where	
$$
	P(x,y,dx,dy) = \left[ (x^2 (a(x) - b(y)) - 1 \right] (dx^2 + dy^2).
$$
Under the assumption (\ref{AH-0}), the metric $g$ is thus in a well-defined sense a small perturbation as $x \to 0$ of the metric
$$
  g_0 = \frac{dx^2 + dy^2}{x^2},
$$
which corresponds to the standard hyperbolic metric on the upper half-plane $\C^+ = \{z = x + iy, \ x > 0\}$. This explains our terminology of \emph{asymptotically hyperbolic ends} for $\{x = 0\}$. By symmetry, the metric is a small perturbation of the standard hyperbolic metric on $\C^+$ when $x \to A$. The Liouville surfaces $\S$ we are considering in this paper thus possess two asymptotically hyperbolic ends at $\{x = 0\}$ and $\{x=A\}$ as well as this remarkable property of separability for the wave equation mentioned above. We refer to \cite{Bor, IK, IKL, JSB} for a detailed presentation of general asymptotically hyperbolic surfaces (and higher dimensional asymptotically hyperbolic manifolds) with no hypothesis of (hidden) symmetries or integrability of the geodesic flow.

We start defining the main object from which we shall study some inverse problems at fixed energy for such asymptotically hyperbolic Liouville surfaces. Precisely, we use the separability of the wave equation to construct a generalized Weyl-Titchmarsh operator as follows. As usual with the models of asymptotically hyperbolic manifolds in \cite{Bor, IK, IKL, JSB}, we consider the shifted stationary wave equation
\begin{equation} \label{ShiftedWaveEq}
  - \triangle_g f - \frac{1}{4} f = \lambda^2 f,
\end{equation}
where $\lambda \not=0$ is a fixed energy. Indeed, it is known (\cite{IKL}) that the essential spectrum of $- \triangle_g$ is $[\frac{1}{4}, +\infty)$ and thus, we shift the bottom of the essential spectrum to $0$ by our choice of wave equation. It is also known that the operator $- \triangle_g - \frac{1}{4}$ has no eigenvalues embedded into the essential spectrum $[0, +\infty)$ (see \cite{Bou, IK, IKL}).

After a simple calculation, the equation (\ref{ShiftedWaveEq}) can be written in our Liouville coordinates as
\begin{equation} \label{SWEq}
  - \partial_x^2 f - \left(\lambda^2 + \frac{1}{4}\right) a(x) f - \partial_y^2 f + \left( \lambda^2 +\frac{1}{4} \right) b(y) f = 0.
\end{equation}
Having in mind the separated equations (\ref{Separated-ODE}), we introduce the angular operator
\begin{equation} \label{AngularOp}
  A(\lambda) = -\frac{d^2}{dy^2} + \left( \lambda^2 +\frac{1}{4} \right) b(y),
\end{equation}
appearing in (\ref{SWEq}), which we view as acting on $ \H_B = L^2((0,B); dy)$. We impose periodic boundary conditions at $\{y=0\}$ and $\{y = B\}$ that make the operator $A(\lambda)$ selfadjoint on $\H_B$. We have

\begin{prop} \label{Angular-Op}
The selfadjoint operator $A(\lambda)$ on $\H_B$ has discrete spectrum with the following properties. There exist a sequence of real eigenvalues $(\mu_{n \lambda}^2)_{n \in \N}$ and normalized eigenfunctions $(Y_{n \lambda})_{n \in \N} \in \H_B$ such that \\
\begin{equation}
  (1) \quad -Y_{n \lambda}'' + \left( \lambda^2 +\frac{1}{4} \right) b(y) Y_{n \lambda} = \mu_{n \lambda}^2 Y_{n \lambda}.
\end{equation}

\begin{equation} \label{HilbertBasis}
  (2) \quad \H_B = L^2((0,B); dy) = \oplus_{n \in \N} \kl Y_{n \lambda} \er.
\end{equation}

\begin{equation} \label{Eigenvalues}
(3) \quad -\infty < \mu_{0 \lambda}^2 \leq \mu_{1 \lambda}^2 \leq \dots \leq \mu_{n \lambda}^2 \leq \dots, \quad \mu_{n \lambda}^2 \to +\infty, \quad n \to \infty.
\end{equation}

\begin{equation} \label{WeylLaw}
(4) \quad \frac{\mu_{n \lambda}^2}{n^2} \sim \frac{\pi^2}{B^2}, \quad n \to \infty, \quad \textrm{Weyl law}.
\end{equation}
\end{prop}

The functions $Y_{n\lambda}$ are thus solutions of the angular ODE
\begin{equation} \label{AngularODE}
  -v''(y) + \left( \lambda^2 +\frac{1}{4} \right) b(y) v(y) = \mu^2 v(y),
\end{equation}
for the particular values $\mu^2 = \mu_{n\lambda}^2$ corresponding to the eigenvalues (counted with multiplicities) of the angular operator $A(\lambda)$ with periodic boundary conditions. By analogy with standard surfaces having a symmetry of revolution, we shall call the numbers $(\mu_{\lambda n})_{n \in \N}$ the \emph{generalized angular momenta}. From (3) above, we only know that the eigenvalues $\mu_{n \lambda}^2$ are positive for $n$ large enough, say $n \geq N$. In particular, for $n \geq N$ we define $\mu_{n \lambda} = \sqrt{\mu_{n \lambda}^2} \in \R^+$ and for $n < N$, we define $\mu_{n \lambda} = i \sqrt{-\mu_{n \lambda}^2} \in i\R^+$. Observe also that the Weyl law implies the following condition of M\"untz type.
\begin{equation}\label{MuntzCond}
  \sum_{n \in \N} \frac{1}{|\mu_{n \lambda}|} = + \infty.
\end{equation}
This condition will be important later in this paper when we use the Complex Angular Momentum (CAM) method.

\begin{rem}
  Note that the above properties (1)-(4) in Proposition \ref{Angular-Op} are satisfied for any Sturm-Liouville operator $A(\lambda)$ with \emph{regular} potential $b \in L^1((0,B);dy)$ together with selfadjoint separated or coupled boundary conditions at $\{y=0\}$ and $\{y=B\}$. We refer to \cite{Ze}, p. 72-73 for a complete description of the different selfadjoint boundary conditions and properties of regular Sturm-Liouville operators. We emphasize that these properties are the only ones needed to construct the generalized Weyl-Titchmarsh operator that will be defined below. %Therefore, we could have considered asymptotically hyperbolic Liouville surfaces $\S$ with boundaries
\end{rem}

Now, let $f$ be a solution of (\ref{SWEq}) belonging to the natural energy Hilbert space on $\S$,
$$
  f \in \H = L^2(\S, dVol_g) = L^2(\S, (a(x) - b(y))dx dy).
$$
Using our assumption on $a(x)$ and $b(y)$, it is not difficult to see that there exists a constant $c$ such that $a(x) - c > 0$ and
\begin{equation} \label{yy1}
  \H = L^2(\S, (a(x) - c) dx dy) = L^2((0,A), ( a(x) - c) dx)) \otimes L^2((0,B), dy).
\end{equation}
Consider then the decomposition
$$
  f(x,y) = \sum_{n \in \N} u_n(x) Y_{n \lambda}(y) \in \H,
$$
given by (\ref{HilbertBasis}) and (\ref{yy1}) above. Putting this decomposition into (\ref{SWEq}), we see that the functions $u_n(x)$ satisfy the radial separated ODE
\begin{equation} \label{RadialODE}
  - u_n''(x) + q(x) u_n(x) = -\mu_{n \lambda}^2 u_n(x),
\end{equation}
with the potential $q$ given by
\begin{equation} \label{q}
  q(x) = \left((i\lambda)^2 - \frac{1}{4}\right) a(x).
\end{equation}
According to the assumptions (\ref{AH-0}) - (\ref{AH-A}) , the potential $q$ has the following asymptotics in the asymptotically hyperbolic ends $x \to 0$ and $x \to A$
\begin{equation} \label{Asymp}
\left\{ \begin{array}{ccc}
  q(x) & = & \frac{(i\lambda)^2 - \frac{1}{4}}{x^2} + q_0(x), \quad \quad x q_0(x) \in L^1(0,\frac{A}{2}), \\
	q(x) & = & \frac{(i\lambda)^2 - \frac{1}{4}}{(A-x)^2} + q_1(x), \quad \quad (A-x) q_1(x) \in L^1(\frac{A}{2},A).
\end{array} \right.
\end{equation}
Hence, the Sturm-Liouville equation (\ref{RadialODE}) has regular singularities at both $x = 0$ and $x=A$. It turns out that this equation enters almost exactly into the framework of the paper \cite{FY} by Freiling and Yurko\footnote{In their paper \cite{FY}, Freiling and Yurko studied different inverse problems for the singular Sturm-Liouville equation $- u''(x) + q(x) u(x) = \rho^2 u(x)$ where the potential $q$ has the asymptotics $$ \left\{ \begin{array}{ccc}
  q(x) & = & \frac{\nu_1^2 - \frac{1}{4}}{x^2} + q_0(x), \quad \quad x^{1-2\Re(\nu_1)} q_0(x) \in L^1(0,\frac{A}{2}), \\
	q(x) & = & \frac{\nu_2^2 - \frac{1}{4}}{(A-x)^2} + q_1(x), \quad \quad (A-x)^{1-2\Re(\nu_2)} q_1(x) \in L^1(\frac{A}{2},A).
\end{array} \right. $$ with $\Re(\nu_1) > 0$ and $\Re(\nu_2) > 0$ and with non-selfadjoint boundary conditions at $x = 0$ and $x=A$.}, in which a generalized Weyl-Titchmarsh function is constructed for such singular Sturm-Liouville equations. We follow here the approach given in \cite{FY} and recall this construction.

We consider the radial separated ODE
\begin{equation} \label{R-ODE}
  - u''(x) + q(x) u(x) = -\mu^2 u(x),
\end{equation}
with $q$ satisfying (\ref{Asymp}). Note that in \cite{FY}, the quantity $\rho^2 = -\mu^2$ is used as the spectral parameter in (\ref{R-ODE}). Even though this is a more usual choice, we choose to work with the original parameter $\mu^2$ which corresponds truly to the generalized angular momentum obtained from the angular ODE. In what follows, it suffices to put $\rho = -i \mu$ to recover the definition and results presented in \cite{FY}.

In Section \ref{SeparatedRadialODE}, we shall prove the existence of fundamental systems of solutions (FSS) $\{S_{10}, S_{20}\}$ and $\{S_{11}, S_{21}\}$ for (\ref{R-ODE}) having the following properties \\

a)
\begin{eqnarray}
  S_{10}(x,\lambda, \mu) \sim C_{10} x^{\frac{1}{2} - i\lambda}, x \to 0, & \quad S_{11}(x,\lambda,\mu) \sim C_{11} (A-x)^{\frac{1}{2} - i\lambda}, x \to A, \label{Jost0} \\
  S_{20}(x,\lambda, \mu) \sim \frac{1}{2i\lambda C_{10}} x^{\frac{1}{2} + i\lambda}, x \to 0, & \quad  S_{21}(x,\lambda,\mu) \sim -\frac{1}{2i\lambda C_{11}} (A-x)^{\frac{1}{2} + i\lambda}, x \to A, \label{JostA}
\end{eqnarray}
where $C_{10}, \ C_{11}$ are arbitrary non-zero constants. \\

b) The relations $W(S_{1m}, S_{2m}) = 1$ hold for $m=0,1$ where $W( f, g )$ denotes the usual Wronskian defined by $W(f,g) = f(x) g'(x) - f'(x) g(x)$. \\

c) For all $x \in (0,A)$ fixed, the functions $ \mu \mapsto S_{jm}(x,\lambda,\mu)$ are entire with respect to $\mu \in \C$. Note that in fact the functions $S_{jm}(x,\lambda,\mu)$ depend on $\mu^2$ and are then even in $\mu$. We shall thus often denote them $S_{jm}(x,\lambda,\mu^2)$. \\

Since the energy $\lambda$ will be fixed in this paper, and will play no role, we shall also often remove the dependence on $\lambda$ of the functions and scattering quantities being consdered. For instance, we shall write $S_{jm}(x,\mu^2)$ instead of $S_{jm}(x,\lambda,\mu^2)$.

We add now some singular separated boundary conditions at $x=0$ and $x=A$ and consider the equation (\ref{R-ODE}) as an eigenvalue problem. Precisely, we consider the separated boundary conditions
\begin{equation} \label{BC}
  U(u) := W(S_{10},u)_{|x=0} = 0, \quad \quad V(u) := W(S_{11},u)_{|x=A} = 0.
\end{equation}
The characteristic function of the above eigenvalue problem is defined by
\begin{equation} \label{Char}
  \Delta_{q}(\mu^2) = W(S_{11}, S_{10}).
\end{equation}
Since $\{S_{11}, S_{10} \}$ are linearly independent solutions of (\ref{R-ODE}), the characteristic function is independent of $x$ and is entire w.r.t. $\mu$. Clearly, the zeros $(\alpha_n)_{n \in \Z}$ of $\Delta_q(\mu^2)$ correspond to the eigenvalues of the equation (\ref{R-ODE})-(\ref{BC}). It can be shown that there is an infinite number of such zeros. By analyticity, they have no accumulation point in $\C$ and thus, $|\alpha_n| \to \infty$ as $|n| \to \infty$.

The generalized Weyl-Titchmarsh (WT) function $M_q(\mu^2)$ is defined by the requirement that the solution of (\ref{R-ODE}) given by
$$
  \phi(x,\mu^2) = S_{20}(x,\mu^2) + M_q(\mu^2) S_{10}(x,\mu^2),
$$
satisfies the boundary condition at $x=A$. Using the relations $W(S_{1m}, S_{2m}) = 1$ for $m=0,1$, we thus obtain the following expression of the generalized Weyl-Titchmarsh function
\begin{equation} \label{WT}
  M_q(\mu^2) = - \frac{W(S_{11}, S_{20})}{W(S_{11}, S_{10})} = - \frac{W(S_{11}, S_{20})}{\Delta_q(\mu^2)}.
\end{equation}

We point out that the above definition generalizes the usual definition of classical Weyl-Titchmarsh functions from \emph{regular} to \emph{singular} Sturm-Liouville (SL) differential operators. When dealing with selfadjoint singular SL operators, the definition and main properties of generalized WT functions can be found in \cite{KST}. In our case, the boundary conditions (\ref{BC}) make the SL equation (\ref{R-ODE}) non-selfadjoint. The generalized WT function can nevertheless be defined by the same recipe as shown in \cite{FY} and recalled above.

Note also that the zeros $(\alpha_n)_{n \in \Z}$ of the characteristic functions are poles of the generalized WT functions. By analogy with the radial Schr\"odinger operator \cite{Ra, Re}, we shall call these zeros $(\alpha_n)_{n \in \Z}$ the \emph{generalized Regge poles} for the stationary wave equation (\ref{ShiftedWaveEq}).

Our interest in considering the generalized WT function $M_q(\mu^2)$ comes from the fact that it is a powerful tool to prove uniqueness results for one-dimensional inverse problems. More precisely, it is well known that the knowledge of the generalized WT function on the complex plane (in fact only on some rays dissecting the complex plane properly \cite{KST}) allows one to determine uniquely the potential $q$ of the equation (\ref{R-ODE}) as well as the boundary conditions (\ref{BC}). We refer to \cite{Be, GS, Te} for results in the case of regular WT functions and to the recent results \cite{FY, KST} in the case of singular WT functions.

The characteristic and generalized WT functions obtained for each one-dimensional equation (\ref{R-ODE}) can be summed up on the spans of each of the harmonics $(Y_{n \lambda})_{n \in \N}$ in order to define operators from $\H_B$ onto itself. Precisely, recalling that
$$
  \H_B = L^2((0,B); dy) = \oplus_{n \in \N} \kl Y_{n \lambda} \er,
$$
we have

\begin{defi} \label{WT-Function}
  Let $\lambda \not=0 $ be a fixed energy. The characteristic operator $\Delta(\lambda)$ and the generalized WT operator $M(\lambda)$ are defined as operators from $\H_B$ onto $\H_B$ that are diagonalizable on the Hilbert basis of eigenfunctions $\{Y_{n \lambda} \}_{n \in \N}$ associated to the eigenvalues $\Delta_q(\mu_{n \lambda}^2)$ and $M_q(\mu_{n \lambda}^2)$. More precisely, for all $v \in \H_B$, $v$ can be decomposed as
$$
  v = \sum_{n \in \N} v_n Y_{n \lambda}, \quad v_n \in \C,
$$
and we have
$$
  \Delta(\lambda) v = \sum_{n \in \N} \Delta_q(\mu_{n \lambda}^2) v_n Y_{n \lambda}, \quad \quad M(\lambda) v = \sum_{n \in \N} M_q(\mu_{n \lambda}^2) v_n Y_{n \lambda}.	
$$
\end{defi}

\begin{rem}
Under our assumptions on the metric, the eigenvalues $\mu_{n\lambda}^2$ of the operator $A(\lambda)$ aren't necessarily simple but can have multiplicity $2$. Hence there is a slight degree of freedom in the choice of the associated eigenfunctions $Y_{n\lambda}$. Nevertheless, we can remark that the definition of the Characteristic and generalized Weyl-Titchmarsh operators do not depend on the choice of the eigenfunctions $Y_{n\lambda}$. This observation will be important later in our study of the inverse problem. 
\end{rem}

We state now the main result of this paper which gives a positive answer to the uniqueness inverse problem for the asymptotically hyperbolic Liouville surfaces we consider.

\begin{thm} \label{Main}
  Let $(\S,g)$ and $(\tilde{\S}, \tilde{g})$ be two asymptotically hyperbolic Liouville surfaces as described in Definition \ref{AHLS}. We shall add a $\ \tilde{}$ to any quantities related to $(\tilde{\S}, \tilde{g})$. Assume that $B = \tilde{B}$. Let $\lambda \not=0$ be a fixed energy of the stationary wave equation (\ref{ShiftedWaveEq}). Assume that the following conditions is fulfilled
$$
  M(\lambda) = \tilde{M}(\lambda).
$$
Then $A = \tilde{A}$ and there exists a constant $C \in \R$ such that
$$
  \tilde{a} = a + C, \quad \quad \tilde{b} = b + C.
$$
As a consequence, we have
$$
  \tilde{g} = g.
$$
In other words, the Liouville surfaces $(\S,g)$ and $(\tilde{\S}, \tilde{g})$ coincide up to isometries.
\end{thm}

Let us make a few comments on this result. First, the hypothesis $B = \tilde{B}$ has the following geometrical interpretation. Recall that the Liouville surface has infinite area in the asymptotically hyperbolic ends $\{x = 0\}$ and $\{x = A\}$. We can quantify how the area in the ends tends to infinity by the following computation. Let $\e > 0$ and consider the part $\S_\e$ of $\S$ given by $\e < x < \frac{A}{2}$ and $0 < y < B$. Then
$$
  \A_\e = Vol_g(S_\e) = \int_\e^{\frac{A}{2}} \int_0^B \left( a(x) - b(y) \right) dx dy.
$$	
A short calculation using (\ref{AH-0}) shows that
$$
  \A_\e \sim \frac{B}{\e}, \quad \e \to 0.
$$
Hence, the constant $B$ appears in the leading term of $\A_\e$ which measures the rate of growth of the area in the asymptotically hyperbolic ends $\{x=0\}$. A similar calculation of course holds for the asymptotically hyperbolic end $\{x=A\}$ where the constant $B$ has the same geometrical interpretation.

Second, our result asserts that asymptotically hyperbolic Liouville surfaces are determined uniquely from the knowledge of the generalized WT operators at a fixed energy \emph{up to isometries}. We stress the fact that generically, the group of isometries of the Liouville surfaces we are considering is at most \emph{discrete} and not \emph{continuous}. Examples of such isometries would be the following. Assume for instance that the function $a$ has the reflection symmetry $a(x) = a(A-x), \ \forall x \in (0,A)$, then the change of variable $(x,y) \mapsto (X = A-x, y)$ would clearly be an isometry since $g = (a(X) - b(y)) (dX^2 + dy^2)$. An example of continuous group of isometries would be if the function $b$ were a constant, \textit{i.e.} $b(y) = b, \ \forall y \in (0,B)$. Then, the Liouville surface would have a cylindrical symmetry with respect to the variable $y$ and thus would be a surface a revolution. This can be seen straightforwardly introducing the radial variable $X = \int \sqrt{a(s)-b} \,ds$ for which the metric $g$ takes the form $g = dX^2 + (a(x(X)) - b) dy^2$.

Third, the characteristic and generalized Weyl-Titchmarsh operators introduced in Definition \ref{WT-Function} can be interpreted as the usual transmission and reflection operators associated to the wave equation for the corresponding asymptotically hyperbolic Liouville surface. More precisely, let us recall the construction of the scattering matrix for general asymptotically hyperbolic surfaces as given by Isozaki, Kurylev, Lassas in \cite{IKL} and adapted to our case.

In our particular model, we have two ends and so we introduce two cutoff functions $\chi_0, \chi_1$ defined by
\begin{equation} \label{Cutoff}
  \chi_j \in C^\infty(\R), \ j=0,1, \quad \chi_0 = 1, \ \textrm{on} \ (0,\frac{A}{4}), \quad \chi_1 = 1 \ \textrm{on} \ (\frac{3A}{4}, A), \quad  \chi_0 + \chi_1 = 1 \ \textrm{on} \ (0,A),
\end{equation}
	in order to separate these two ends. It is shown in \cite{IKL} that the solutions of the shifted stationary equation (\ref{ShiftedWaveEq}) are unique when imposing some radiation conditions at infinities. Precisely, we define some Besov spaces that encode these radiation conditions at infinities as follows.

\begin{defi} \label{AbstractBesov}
  Let $\H_B = L^2((0,B);dy) $. Let the intervals $(0,+\infty)$ and $(-\infty, A)$ be decomposed as
	$$
	  (0,+\infty) = \cup_{k \in \Z} I_k, \quad (-\infty, A) = \cup_{k \in \Z} J_k,
	$$
	where
	$$
	  I_k = \left\{ \begin{array}{cc}
		(exp(e^{k-1}), exp(e^k)], & k \geq 1, \\
		(e^{-1}, e], & k = 0, \\
		(exp(-e^{|k|}), exp(-e^{|k|-1})], & k \leq -1,
		\end{array} \right.,
		$$
		$$
		J_k = \left\{ \begin{array}{cc}
		(A-exp(e^{k-1}), A-exp(e^k)], & k \geq 1, \\
		(A-e^{-1}, A-e], & k = 0, \\
		(A-exp(-e^{|k|}), A-exp(-e^{|k|-1})], & k \leq -1.
		\end{array} \right.
	$$
	We define the Besov spaces  $\B_0 = \B_0(\H_B), \B_1 = \B_1(\H_B)$ to be the Banach spaces of $\H_B$-valued functions on $(0,\infty)$ and $(-\infty, A)$ satisfying respectively
	$$
	  \| f \|_{\B_0} = \sum_{k \in \Z} e^{\frac{|k|}{2}} \left( \int_{I_k} \| f(x) \|^2_{\H_B} \frac{dx}{x^2} \right)^{\frac{1}{2}} < \infty,
	$$
	$$
	  \| f \|_{\B_1} = \sum_{k \in \Z} e^{\frac{|k|}{2}} \left( \int_{J_k} \| f(x) \|^2_{\H_B} \frac{dx}{(A-x)^2} \right)^{\frac{1}{2}} < \infty.
	$$
	The dual spaces $\B_0^*$ and $\B_1^*$ are then identified with the spaces equipped with the norms
	$$
	  \| f \|_{\B_0^*} = \left( \sup_{R > e} \frac{1}{\log R} \int_{\frac{1}{R}}^R \| f(x) \|^2_{\H_B} \frac{dx}{x^2} \right)^{\frac{1}{2}} < \infty,
	$$
	$$
	  \| f \|_{\B_1^*} = \left( \sup_{R > e} \frac{1}{\log R} \int_{A-R}^{A-\frac{1}{R}} \| f(x) \|^2_{\H_B} \frac{dx}{(A-x)^2} \right)^{\frac{1}{2}} < \infty.
	$$
\end{defi}

\begin{rem}
  To give a better idea of the nature of these Besov spaces, we may compare them to the more familiar weighted $L^2$ spaces. If we define the spaces $L_0^{2,s}((0,\infty); \H_B)$ for $s \in \R$ by
	$$
	  \|f \|_s = \left( \int_0^\infty (1 + |\log x|)^{2s} \| f(x) \|^2_{\H_B} \frac{dx}{x^2} \right)^{\frac{1}{2}} < \infty,
	$$
  then for $s > \half$, we have
		$$
		 L^{2,s}_0 \subset \B_0 \subset L^{2,\half}_0 \subset L^2_0 \subset L^{2,-\half}_0 \subset \B_0^* \subset L^{2,-s}_0.
		$$
	There is of course a similar description of the Besov spaces $\B_1$ and $\B_1^*$. 	
\end{rem}

\begin{defi} \label{Besov}
  We define the Besov spaces $\B$ and $\B^*$ as the Banach spaces of $\H_B$-valued functions on $(0,A)$ with norms
	$$
	  \| f \|_\B = \| \chi_0 f \|_{\B_0} + \| \chi_1 f \|_{\B_1},
	$$
	and
	$$
	  \| f \|_{\B^*} = \| \chi_0 f \|_{\B_0^*} + \| \chi_1 f \|_{\B_1^*}.
	$$
	We also define the scattering Hilbert space
  $$
    \H_\infty = \H_B \otimes \C^2 \simeq \H_B \oplus \H_B.
  $$
\end{defi}
Then Isozaki, Kurylev and Lassas proved \cite{IKL}

\begin{thm} \label{Stat-Sol}
  a) For any solution $f \in \B^*$ of the shifted stationary wave equation at non-zero energy $\lambda^2$ (\ref{ShiftedWaveEq}), there exists a unique $\psi^{(\pm)} = (\psi_0^{(\pm)}, \psi_1^{(\pm)}) \in \H_\infty$ such that
\begin{eqnarray} \label{Asymp-StatSol}
  	f \simeq & \omega_-(\lambda) \left( \chi_0 \,x^{\half + i\lambda} \psi_0^{(-)} + \chi_1 \,(A-x)^{\half + i\lambda} \psi_1^{(-)} \right) \\
		         & - \omega_+(\lambda) \left( \chi_0 \, x^{\half - i\lambda} \psi_0^{(+)} + \chi_1 \, (A-x)^{\half - i\lambda} \psi_1^{(+)} \right), \nonumber
\end{eqnarray}
where
\begin{equation} \label{Omega}
  \omega_\pm(\lambda) = \frac{\pi}{2 \lambda \sinh(\pi \lambda))^\half \Gamma(1 \mp i \lambda)}.
\end{equation}
b) For any $\psi^{(-)} \in \H_\infty$, there exists a unique $\psi^{(+)} \in \H_\infty$ and $f \in \B^*$ satisfying (\ref{ShiftedWaveEq}) for which the decomposition (\ref{Asymp-StatSol}) above holds. This defines uniquely the scattering operator $S(\lambda)$ as the $\H_\infty$-valued operator such that for all $\psi^{(-)} \in \H_\infty$
\begin{equation}
  \psi^{(+)} = S(\lambda) \psi^{(-)}.
\end{equation}
c) The scattering operator $S(\lambda)$ is unitary on $\H_\infty$.
\end{thm}

Note that the above scattering operator has the structure of a $2 \times 2$ matrix whose components are $\H_B$-valued operators. Precisely, we write
$$
  S(\lambda) = \left[ \begin{array}{cc} L(\lambda) & T_R(\lambda) \\
	                                      T_L(\lambda) & R(\lambda)
											\end{array} \right], 									
$$
where $T_L(\lambda), T_R(\lambda)$ are the transmission operators and $L(\lambda), R(\lambda)$ are the reflection operators from the right and from the left. Their interpretation is that the transmission operators measure what is transmitted from one end to the other in a scattering experiment, while the reflection operators measure the part of a signal sent from one end that is reflected to itself.

Our last result is the fact that the characteristic and Weyl-Titchmarsh operators of Definition \ref{WT-Function} are nothing but the inverse of the transmission operator for the former and the reflection operator for the latter. Precisely, we shall prove

\begin{prop} \label{WT-Reflection}
  Let $\lambda \not=0$ be a fixed energy. Then we have
\begin{eqnarray*}
	  \Delta(\lambda) & = & \frac{2i\lambda C_{10} C_{11} \Gamma(1-i\lambda)}{\Gamma(1+i\lambda)} \left( T(\lambda) \right)^{-1}, \\
	  M(\lambda)      & = & - \frac{\Gamma(1+i\lambda)}{2i\lambda C_{10}^2 \Gamma(1-i\lambda)} L(\lambda) = 
                            \frac{|C_{11}|^2 \Gamma(1+i\lambda) }{2i\lambda C_{11}^2  | C_{10}|^2\Gamma(1-i\lambda)} R(\lambda),
\end{eqnarray*}	
where $C_{10}, C_{11}$ are the arbitrary constants that appear in (\ref{Jost0}) - (\ref{JostA}).
\end{prop}

Hence our main Theorem \ref{Main} could be rephrased  as follows. \emph{The knowledge of the reflection operators $L(\lambda)$ or $R(\lambda)$ and thus, of the full scattering matrix, at a fixed non-zero energy $\lambda$ determines uniquely an asymptotically hyperbolic Liouville surface up to isometries}. This gives a complete answer to the uniqueness part in the inverse scattering problem at fixed energy for such surfaces. In the case of general (without hidden symmetry) asymptotically hyperbolic surfaces or higher dimensional asymptotically hyperbolic manifolds, we refer to \cite{IK, IKL, JSB, SB} for uniqueness results from the knowledge of the scattering matrix \emph{at all energies}.

We conclude this introduction saying a few words on the strategy of the proof of Theorem \ref{Main} from our main assumption $M(\lambda) = \tilde{M}(\lambda)$ on the WT operators at a fixed energy $\lambda^2$. We recall that we add a $\,\tilde{}$ at any quantities related to the Liouville surface $(\tilde{\S}, \tilde{g})$. By construction of the WT operator as a diagonalizable operator on the Hilbert basis $\{Y_{n \lambda}\}$ associated to the eigenfunctions $M_q(\mu_{n\lambda}^2)$, this hypothesis will roughly speaking imply that
\begin{equation} \label{1}
  Y_{n \lambda}(y) = \tilde{Y}_{n\lambda}(y), \quad \forall y \in (0,B), \ \forall n \in \N,
\end{equation}
and
\begin{equation} \label{2}
  M_q(\mu_{n\lambda}^2) = M_{\tilde{q}}(\tilde{\mu}_{n\lambda}^2), \quad \forall n \in \N.
\end{equation}

Recalling that the functions $Y_{n\lambda}$ are eigenfunctions of the angular separated equation (\ref{AngularODE}), we can use (\ref{1}) in a direct way to conclude that there exists a constant $C$ such that
$$
  \tilde{b}(y) = b(y) + C, \quad \forall y \in (0,B),
$$
and
$$
  \tilde{\mu}^2_{n\lambda} = \mu^2_{n\lambda} + C (\lambda^2 +\frac{1}{4}), \quad \forall n \in \N.
$$
In other words, the angular separated equation together with our main assumption allow us to recover the potential $b$ (up to a constant). For simplicity, we assume in this introduction that this constant $C$ vanishes. Thus the condition (\ref{2}) becomes
\begin{equation} \label{3}
  M_q(\mu_{n\lambda}^2) = M_{\tilde{q}}(\mu_{n\lambda}^2), \quad \forall n \in \N.
\end{equation}

The second step is to use the Complex Angular Momentum (CAM) method, that is to say to allow the angular momentum $\mu_{n\lambda}$ to be complex - denoted generically by $\mu$ - and to use some uniqueness results for analytic functions to extend the validity of (\ref{3}) to all $\mu \in \C \setminus \{ \textrm{poles} \}$. We shall show in Section \ref{SeparatedRadialODE} that the numerator and denominator of the function $\mu \mapsto M_q(\mu^2)$ are entire in $\mu \in \C$ and even better, belong to the Cartwright class of entire functions. Now, it is well known that such functions are uniquely determined by their values on a sequence of complex number $(z_n)_n$ satisfying a M\"untz condition of the following type
$$
  \sum_{n \in \N} \frac{1}{|z_n|} = \infty.
$$
But this is precisely the case of the eigenvalues $\mu_{n \lambda}$ that grows as the integers as $n\to \infty$ according to (\ref{WeylLaw}). This will help us to prove that if the equality (\ref{3}) holds on all $\mu_{n \lambda}$, then it holds for all $\mu \in \C \setminus \{ \textrm{poles} \}$.

The last step of the proof is an application of the theory of generalized Weyl-Titchmarsh functions in one-dimensional inverse problems. Recalling that $M_q(\mu^2)$ is the generalized WT function associated to the separated radial ODE (\ref{R-ODE}), it is an almost standard procedure to show that we can recover the potential $q$ - and thus the function $a$ - from the knowledge of $M_q(\mu^2)$ for all $\mu \in \C$. This third step finishes the proof of our main Theorem \ref{Main} since the metric $g$ only depends on $a - b$.

In the remainder of this paper, we shall provide the details of the results stated above. In Section \ref{LiouvilleSurfaces}, we come back to the definition and properties of Liouville surfaces. We explain in which sense each Riemanniann surface possessing this property of separability of the wave equation can be written locally as a Liouville surface. In Section \ref{SeparatedRadialODE}, we construct the FSS for the separated radial ODE. We provide estimates on these FSS and also on the scattering coefficients $\Delta_q(\mu^2)$ and $M_q(\mu^2)$ as $|\mu| \to \infty$. In Section \ref{IP}, we solve the inverse problems in three main steps as explained above. Finally, we finish this paper with an open problem, namely can we determine the metric from the transmission operators at a fixed energy? We provide some preliminary remarks on this problem and conjecture a result of non uniqueness.

%At last, in appendix \ref{WT-DN}, we study the case of compact Liouville surfaces with boundaries. We define in this setting the generalized WT operator as we did for asymptotically hyperbolic surfaces and show that it coincides with the Dirichlet-to-Neumann map. Then we recover the uniqueness (up to isometries) result for the class of Liouville metrics which is of course not new in 2D. But we found interesting to make the link between a priori different objects in the litterature.

%%%%%%%%%%%%%%%%%%%%%%%%%%%%%%%%%%%%%%%%%%%%%%%%%%%% ASYMPTOTICALLY HYPERBOLIC STAECKEL SURFACES %%%%%%%%%%%%%%%%%%%%%%%%%%%%%%%%%%%%%

\Section{The wave equation on Liouville surfaces. Separation of variables.} \label{LiouvilleSurfaces}

Our purpose in this section is to describe and characterize geometrically the class of manifolds in which we shall solve the inverse scattering problem at fixed energy for the wave equation.
Specifically, we consider three-dimensional manifolds $\M$ which we assume to be diffeomorphic to the product of the real line $\mathbb{R}_{t}$ with a surface $\S$ which is diffeomorphic to a cylinder.  We endow the manifold $\M$ with a Lorentzian product metric $g_{\M}$, given by
\[
g_{\M}=-dt^{2}+g_{\S},
\]
where $g_{\S}$ is a Riemannian metric on $\S$. We shall impose on the Riemannian surface $(M,g_{\S})$ that it be asymptotically hyperbolic, and that there exist a global coordinate chart on $\S$ such that the wave equation
\begin{equation}\label{wave3}
\Box_{g_{\M}}\psi=0,
\end{equation}
is separable into ordinary differential equations when $\psi$ is expressed as a product of functions on one variable. There is still quite a bit of freedom left in choosing coordinates metrics satisfying the above hypotheses, and the Liouville form arises naturally when trying to construct a simple canonical form in which all the requirements are met. We begin with the separability condition, which is local in nature.

We first observe that in view of the product structure of $g_{\M}$, the dependence of $\psi$ on the time variable $t$ can be separated by setting
\begin{equation}
\psi=e^{i\lambda t}f,
\end{equation}
where $\lambda$ is a separation constant and $f$ is independent of $t$. The equation~(\ref{wave3}) then reduces to
\begin{equation}\label{Helm}
-\Delta_{g}f=\lambda^{2}f,
\end{equation}
which is the eigenvalue equation for the Laplace-Beltrami operator on the Riemannian surface $(M,g_{\S})$, also called the Helmholtz equation.

The $n$-dimensional metrics admitting local coordinates in which the Helmholtz equation is separable into ordinary differential equations have been studied extensively, both in the orthogonal case~\cite{Eisen} and in the non-orthogonal case~\cite{KM}. These works have their origin in the pioneering investigations of St\"ackel on the separability of the Hamilton-Jacobi equation for the geodesic flow of a Riemannian metric admitting orthogonal coordinates, and are anchored around the notion of metrics in St\"ackel form, which we now recall. Consider an open subset $U$ of $\mathbb{R}^{n}$ endowed with a metric for which the coordinate lines are orthogonal,
\begin{equation}\label{orth}
ds^{2}=\sum_{i=1}^{n}H_{i}^{2}(x^{1},\ldots,x^{n})(dx^{i})^{2}.
\end{equation}
We say that the metric is in \emph{St\"ackel form} if there exists a matrix $S=(s_{ij})$ of $C^{\infty}$ functions on $U$, called a St\"ackel matrix, such that
\[
H_{i}^{2}=\frac{s}{s^{i1}},
\]
where
\[
\frac{\partial s_{ij}}{\partial x^{k}}=0\quad \mbox {for} \quad k\neq i,
\]
where $s:=\det S$ and $s^{i1}$ denotes the cofactor of $s_{i1}$. It is well-known that the Hamilton-Jacobi equation for the geodesic flow of the metric~(\ref{orth})
\begin{equation}\label{HJ}
\sum_{i=1}^{n}H_{i}^{-2}(\frac{\partial W}{\partial x^{i}})^{2}=E,
\end{equation}
admits an additively separable complete integral if and only if the metric is in St\"ackel form~\cite{Eisen}. It is proved in~\cite{Eisen} that necessary and sufficient conditions for the product separability of the Helmholtz equation~(\ref{Helm}) in the orthogonal metric~(\ref{orth}) are that the metric be in St\"ackel form and furthermore that the {\it Robertson conditions}
\[
s=\prod_{i=1}^{n}\frac{H_{i}}{\psi_{i}},
\]
where the $\psi_{i}$ are $C^{\infty}$ functions on $U$ such that
\[\frac{\partial \psi_{i}}{\partial x^{j}}=0\quad \mbox{for} \quad 1\leq i\neq j \leq n,
\]
be satisfied. It is proved in \cite{Eisen} that the Robertson conditions are equivalent to
\[
R_{ij}=0 \quad \mbox{for} \quad 1\leq i\neq j \leq n,
\]
where $R_{ij}$ denotes the Ricci tensor. In the case $n=2$ of surfaces, it is easily verified that \emph{both} the Hamilton-Jacobi equation~(\ref{HJ}) and the Helmholtz equation~(\ref{Helm}) are separable if the metric is in St\"ackel form, in other words that the Robertson condition is vacuous. Furthermore, it is also straightforward to verify~\cite{Mor} that every separable coordinate on an open set of a  Riemannian system is necessarily orthogonal. We can therefore assume without loss of generality that the metric on our Riemannian surface $(\S,g_{\S})$ is locally given in St\"ackel form as
\begin{equation}\label{Stackel2}
ds^{2}=\frac{s_{11}s_{22}-s_{12}s_{21}}{s_{22}}(dx^{1})^{2}+\frac{s_{11}s_{22}-s_{12}s_{21}}{s_{12}}(dx^{2})^{2}.
\end{equation}

We are now going to implement our global assumptions about the surface $(\S,g_{\S})$, namely that it is diffeomorphic to a cylinder, with asymptotically hyperbolic ends. The first step will be to choose a specific system of coordinates within the class of coordinates in which the metric is in the St\"ackel form (\ref{Stackel2}). Indeed, we may use the freedom to redefine our separable coordinates $(x^{1},x^{2})$ by a local diffeomorphism of the form $(x^{1},x^{2})\mapsto (x(x^{1}), y(x^{2}))$, where the functions $x,y$ are $C^{\infty}$ and monotone\footnote{Just define $$x= \int \sqrt{s_{12}} dx^1, \quad y= \int \sqrt{s_{22}} dx^2.$$}, to put the metric (\ref{Stackel2}) in \emph{Liouville form},
\begin{equation}
ds^{2}=(a(x)-b(y))(dx^{2}+dy^{2}).
\end{equation}

By the preceding discussion, we know that  the wave equation (\ref{wave3}) will admit product separable solutions of the form
\begin{equation}\label{sep}
\psi=e^{i\lambda t}u(x)v(y).
\end{equation}
The second step, which can be thought of as adapting the coordinates to the cylindrical topology of our surface $\S$, is to assume that  $\S$ is covered with a single coordinate chart $(0,A)_x \times (0,B)_y$, where the coordinate $y$ is to be thought of as an angular variable defined on a circle of length $2\pi B$, and the coordinate $x$ which can be thought of as the radial variable for our problem, defines the asymptotic ends of the cylinder as $\{x=0\}$ and $\{x=A\}$.  The final step is to require from the functions $a(x)$ and $b(y)$ that their asymptotic behaviour near  $\{x=0\}$ and $\{x=A\}$ correspond to the boundary behaviour of the metric of the hyperbolic plane, with precise error estimates. These are precisely the conditions appearing in Definition \ref{AHLS}. These behaviour of these error terms will be crucial to our analysis of the inverse problem for Liouville surfaces at fixed energy.

%%%%%%%%%%%%%%%%%%%%%%%%%%%%%%%%%%%%%%%%%%%%%%%%% THE DIRECT SCATTERING %%%%%%%%%%%%%%%%%%%%%%%%%%%%%%%%%%%%%%%%%%%%%%%%%%%%%%%%%%

\Section{The separated radial ODE} \label{SeparatedRadialODE}

In this Section, we investigate the separated radial ODE (\ref{R-ODE})
$$
  -u''(x) + q(x) u(x) = -\mu^2 u(x),
$$
where the potential $q$ is given by
$$
  q(x) = \left((i\lambda)^2 - \frac{1}{4}\right) a(x),
$$
and satisfies the asymptotics (\ref{Asymp}). In particular, we construct two fundamental systems of solutions (FSS) of solutions $\{S_{10}, S_{20}\}$ and $\{S_{11}, S_{21}\}$ for (\ref{R-ODE}) having the properties a) - c) of Section \ref{Intro}. 

\subsection{The construction of the Fundamental Systems of Solutions (FSS)}

We shall only construct the FSS $\{S_{10}, S_{20}\}$ that satisfies the asymptotics (\ref{Jost0}) at $x = 0$ since the construction of $\{S_{11}, S_{21}\}$ that satisfies the same asymptotics at $\{x = A\}$ is  identical by symmetry. We rewrite (\ref{R-ODE}) as the inhomogeneous ordinary differential equation
\begin{equation} \label{Inhom}
  -u''(x) - \frac{\lambda^2 + \frac{1}{4}}{x^2} \,u(x) + \mu^2 u(x) = -q_{0}(x) u(x),
\end{equation}
where $x q_{0}(x) \in L^1(0,\frac{A}{2})$ according to (\ref{Asymp}). We first study the homogeneous equation
\begin{equation} \label{Hom}
  - u''(x) - \frac{\lambda^2 + \frac{1}{4}}{x^2} u(x) + \mu^2 u(x) = 0,
\end{equation}
in which we recognize a modified Bessel equation. It is well known that (\ref{Hom}) has a FSS of solutions given by $\{ \sqrt{x} I_{i\lambda}(\mu x), \sqrt{x} K_{i\lambda}(\mu x) \}$, where the modified Bessel functions $I_\nu(z)$ and $K_\nu(z)$ are defined by (see for instance \cite{Leb})
\begin{equation}\label{defI}
 I_{\nu}(z) = \sum_{k=0}^{\infty} \frac{\left( \frac{z}{2} \right)^{\nu + 2k}}{\Gamma(k+\nu+1) k!}, \quad \vert z \vert < \infty, \quad \vert \arg(z) \vert < \pi,
\end{equation}
\begin{equation}\label{defK}
  K_{\nu}(z) = \frac{\pi}{2} \frac{I_{-\nu}(z) - I_{\nu}(z)}{\sin{\nu \pi}}, \quad \vert \arg(z) \vert < \pi, \quad \nu \notin \mathbb{Z}.
\end{equation}
Recall that $W(\sqrt{x} I_{i\lambda}(\mu x), \sqrt{x} K_{i\lambda}(\mu x) ) = -1$ by \cite{Leb}, eq. (5.9.5). We can thus define the Green's function associated to (\ref{Hom}) by the kernel
\begin{equation} \label{Green}
  G(x,t, \mu) = \sqrt{xt} \left( I_{i\lambda}(\mu x) K_{i\lambda}(\mu t) - I_{i\lambda}(\mu t) K_{i\lambda}(\mu x) \right).
\end{equation}
We shall use the following estimates on the Green's kernel:

\begin{prop} \label{GreenEst}
  For all $\mu \in \C$ such that $\Re(\mu) \geq 0$, we have
	\begin{equation} \label{Green1}
    |G(t,x,\mu)| \leq C \left( \frac{x}{1+|\mu| x}\right)^{\frac{1}{2}} \left( \frac{t}{1+|\mu| t}\right)^{\frac{1}{2}} e^{ \Re(\mu) (x-t)}, \quad \forall \ 0 \leq t \leq x,
	\end{equation}
	\begin{equation} \label{Green2}
	  |\partial_x G(t,x,\mu)| \leq C \left( \frac{1+|\mu| x}{x} \right)^{\frac{1}{2}} \left( \frac{t}{1+|\mu| t} \right)^{\frac{1}{2}} e^{ \Re(\mu) (x-t)}, \quad \forall \ 0 \leq t \leq x.
	\end{equation}
\end{prop}
\begin{proof}
  We need the following estimates on the modified Bessel functions $I_\nu(z), K_\nu(z)$ when $\Re(\nu) = 0, \ \nu \ne 0$ (\cite{IK, Leb, Ol}).
	For $\nu \in \C^*, \ \Re(\nu) = 0$, we have for $\Re(z) \geq 0$
	\begin{equation} \label{Est-I}
	  |I_\nu(z)| \leq C (1+|z|)^{-\half} e^{\Re(z)},
	\end{equation}
	\begin{equation} \label{Est-K}
	  |K_\nu(z)| \leq C (1+|z|)^{-\half} e^{-\Re(z)},
	\end{equation}
	$$
	  |zI'_\nu(z)| \leq C (1+|z|)^{\half} e^{\Re(z)},
	$$
	$$
	  |zK'_\nu(z)| \leq C (1+|z|)^{\half} e^{-\Re(z)},
	$$
	where the constant $C$ does not depend on $z, \ \Re(z) \geq 0$. The estimates on the Green's kernel follow by a direct calculation.
\end{proof}

We now come back to the inhomogenous equation (\ref{Inhom}). From the previous discussion, the solutions $S_{j0}, \ j=1,2$ of (\ref{Inhom}) we are constructing must satisfy the integral equation
\begin{equation} \label{Int}
  S_{j0}(x,\mu) = \alpha_j \sqrt{x} I_{i\lambda}(\mu x) + \beta_j \sqrt{x} K_{i\lambda}(\mu x) - \int_0^x G(x,t,\mu) q_{0}(t) S_{j0}(t,\mu) dt,
\end{equation}
for a certain choice of constants $\alpha_j, \beta_j$. Let us solve (\ref{Int}) by iteration, \textit{i.e.} we construct solutions under the form
\begin{equation} \label{IterSerie}
  S_{j0}(x,\mu) = \sum_{k=0}^\infty g_{j k}(x,\mu),
\end{equation}
where
$$
  g_{j0}(x,\mu) = \alpha_j \sqrt{x} I_{i\lambda}(\mu x) + \beta_j \sqrt{x} K_{i\lambda}(\mu x),
$$
and
$$
  g_{j, k+1}(x,\mu) = -\int_0^x G(x,t,\mu) q_0(t) g_{j,k}(t,\mu) dt.
$$

Let us first determine the constants $\alpha_j, \beta_j$ for $j = 1,2$. Recall from (\ref{Jost0}) that we want
$$
  S_{10}(x,\mu) \sim_{x \to 0} C_{10} x^{\frac{1}{2} - i\lambda}, \quad \quad S_{20}(x,\mu) \sim_{x \to 0} \frac{1}{2i\lambda C_{10}} x^{\frac{1}{2} + i\lambda}.
$$
Hence we impose the conditions
\begin{equation} \label{Goal}
  g_{10}(x,\mu) \sim_{x \to 0} C_{10} x^{\frac{1}{2} - i\lambda}, \quad \quad g_{20}(x,\mu) \sim_{x \to 0} \frac{1}{2i\lambda C_{10}} x^{\frac{1}{2} + i\lambda}.
\end{equation}
But we know from (\ref{defI}) that
\begin{equation} \label{I-0}
  \sqrt{x} I_{i\lambda}(\mu x) \sim_{x \to 0} \frac{\mu^{i\lambda}}{2^{i\lambda} \Gamma(1 + i\lambda)} x^{\frac{1}{2} + i\lambda},
\end{equation}
\begin{equation} \label{K-0}
  \sqrt{x} K_{i\lambda}(\mu x) \sim_{x \to 0} \frac{\pi}{2i \sinh(\lambda \pi)} \left[ \frac{\mu^{-i\lambda}}{2^{-i\lambda} \Gamma(1 - i\lambda)} x^{\frac{1}{2} - i\lambda}  - \frac{\mu^{i\lambda}}{2^{i\lambda} \Gamma(1 + i\lambda)} x^{\frac{1}{2} + i\lambda} \right].
\end{equation}
Since $g_{10}(x,\mu) = \alpha_1 \sqrt{x} I_{i\lambda}(\mu x) + \beta_1 \sqrt{x} K_{i\lambda}(\mu x) $, it is easy to see from (\ref{I-0}) and (\ref{K-0}) that we must choose
\begin{equation} \label{c1}
  \alpha_1 = C_{10} \Gamma(1 - i\lambda) \left( \frac{\mu}{2} \right)^{i\lambda}, \quad \beta_1 = C_{10} \frac{2i \sinh(\lambda \pi)}{\pi} \Gamma(1 - i\lambda) \left( \frac{\mu}{2} \right)^{i\lambda},
\end{equation}
in order to get (\ref{Goal}).

Similarly, we must impose
\begin{equation} \label{c2}
  \alpha_2 = \frac{1}{2i\lambda C_{10}} \Gamma(1 + i\lambda) \left( \frac{\mu}{2} \right)^{-i\lambda}, \quad \beta_2 = 0,
\end{equation}
in order to get (\ref{Goal}).

In conclusion, we choose as initial functions of our iterated series
\begin{eqnarray}
   g_{10}(x,\mu) & = & C_{10} \Gamma(1 - i\lambda) \left( \frac{\mu}{2} \right)^{i\lambda} \left[ \sqrt{x} I_{i\lambda}(\mu x) + \frac{2i \sinh(\lambda \pi)}{\pi} \sqrt{x} K_{i\lambda}(\mu x) \right]  \nonumber \\
                 & =&  C_{10} \Gamma(1 - i\lambda) \left( \frac{\mu}{2} \right)^{i\lambda} \sqrt{x} I_{-i\lambda}(\mu x), \label{g10}\\
   g_{20}(x,\mu) & = & \frac{1}{2i\lambda C_{10}} \Gamma(1 + i\lambda) \left( \frac{\mu}{2} \right)^{-i\lambda} \sqrt{x} I_{i\lambda}(\mu x). \label{g20}	
\end{eqnarray}

To prove the convergence of the iterated series (\ref{IterSerie}) and obtain useful estimates as $|\mu| \to \infty$, we shall use

\begin{lemma} \label{IterEst}
  For all $x \in (0,A)$ and for all $\mu \in \C$ such that $\Re(\mu) \geq 0$, for $j = 1,2$, we have
	$$
	  |g_{j0}(x,\mu)| \leq C \left( \frac{x}{1+|\mu| x}\right)^{\frac{1}{2}}  e^{ \Re(\mu) x}.
	$$
\end{lemma}
\begin{proof}
  This is a direct consequence of (\ref{Est-I})-(\ref{Est-K}) and (\ref{g10})-(\ref{g20}).
\end{proof}

From Proposition \ref{GreenEst} and Lemma \ref{IterEst}, it follows by an easy induction that, for all $x \in (0,A)$ and for all $\mu \in \C$ such that $\Re(\mu) \geq 0$, for $j = 1,2$, we have

\begin{equation} \label{IterEst-k}
  |g_{jk}(x,\mu)| \leq \frac{C^{k+1}}{k!} \left( \frac{x}{1+|\mu| x}\right)^{\frac{1}{2}}  e^{ \Re(\mu) x} \left( \int_0^x \frac{t |q_0(t)|}{1+|\mu|t} dt \right)^k, \quad \forall \ k \in \N, \ j=1,2.
\end{equation}
Hence, the iterated serie (\ref{IterSerie}) converges and satisfies (\ref{Int}). Moreover, for all $x \in (0,A)$ and for all $\mu \in \C$ such that $\Re(\mu) \geq 0$, for $j = 1,2$, we also have the estimates

\begin{equation} \label{Est-Sj0}
  |S_{j0}(x,\mu)| \leq C \left( \frac{x}{1+|\mu| x}\right)^{\frac{1}{2}}  e^{ \Re(\mu) x} e^{C \left( \int_0^x \frac{t |q_0(t)|}{1+|\mu|t} dt \right)}, \quad j=1,2,
\end{equation}
and
\begin{equation} \label{Est-Sj0-1}
  |S_{j0}(x,\mu) - g_{j0}(x,\mu)| \leq C \left( \frac{x}{1+|\mu| x}\right)^{\frac{1}{2}}  e^{ \Re(\mu) x} \left[ e^{C \left( \int_0^x \frac{t |q_0(t)|}{1+|\mu|t} dt \right)} - 1 \right], \quad j=1,2.
\end{equation}

Recall now that by (\ref{AH-0}) with $\alpha = 0$ and $n = 0$, we have for all $x \in (0, x_0)$, where $x_0 \in (0,A)$ is fixed,
\begin{equation} \label{Est-q0}
  |x q_0(x) | \leq \frac{C}{x (1 + |\log x |)^{1 + \eo}}, \quad \eo > 0.
\end{equation}
Thus we obtain for all $x\sqrt{|\mu|} \geq 1$ such that $|\mu| \geq 1$ and $\Re(\mu) \geq 0$
\begin{eqnarray*}
  \int_0^x \frac{t |q_0(t)|}{1+|\mu|t} dt & = & \int_0^{\frac{1}{\sqrt{|\mu|}}} \frac{t |q_0(t)|}{1+|\mu|t} dt  + \int_{\frac{1}{\sqrt{|\mu|}}}^x \frac{t |q_0(t)|}{1+|\mu|t} dt, \\
	& \leq & C \int_0^{\frac{1}{\sqrt{|\mu|}}} \frac{dt}{t |\log t|^{1 + \eo}} dt + \int_{\frac{1}{\sqrt{|\mu|}}}^x \frac{|t q_0(t)|}{|\mu|t} dt , \\
	& \leq & \frac{C}{\left( - \log \frac{1}{\sqrt{|\mu|}} \right)^{\eo}} + \frac{1}{\sqrt{|\mu|}} \int_{\frac{1}{\sqrt{|\mu|}}}^x |t q_0(t)| dt, \\
	& \leq & C \left[ \frac{1}{(\log|\mu|)^{\eo}} + \frac{1}{\sqrt{|\mu|}} \right], \\
	& \leq & \frac{C}{(\log|\mu|)^{\eo}},
\end{eqnarray*}
where the constants $C$ may differ from line to line. As a consequence, we get the following estimates as $|\mu|\to \infty, \ \Re(\mu) \geq 0$ and $x|\mu| \geq 1$.
\begin{equation} \label{Est-mu1}
  \int_0^x \frac{t |q_0(t)|}{1+|\mu|t} dt = O \left( \frac{1}{(\log|\mu|)^{\eo}} \right),
\end{equation}
\begin{equation} \label{Est-mu2}
  e^{C \int_0^x \frac{t |q_0(t)|}{1+|\mu|t} dt} - 1 = O \left( \frac{1}{(\log|\mu|)^{\eo}} \right),
\end{equation}

We thus infer from (\ref{Est-Sj0}), (\ref{Est-Sj0-1}), (\ref{Est-mu1}) and (\ref{Est-mu2}) that, given a $x_0 \in (0,A)$ fixed, we have for all $x \in (0,x_0)$ and for $j=1,2$, as $|\mu| \to \infty$ with $\Re(\mu) \geq 0$
  $$
	  |S_{j0}(x,\mu)| \leq C \frac{e^{\Re(\mu) x}}{|\mu|^\half},
	$$
	$$
	  |S_{j0}(x,\mu) - g_{j0}(x,\mu)| \leq C \frac{e^{\Re(\mu) x}}{|\mu|^\half \, \log|\mu|^{\eo}}.
	$$

In order to compute the necessary Wronskians of the FSS $\{S_{10}, S_{20}\}$ and $\{S_{11}, S_{21}\}$, we shall need good estimates on the derivatives $S'_{jm}, \ j=1,2, \ m=0,1$. Since $G(x,x,\mu) = 0$ by (\ref{Green}), we easily see that the derivatives $S'_{j0}$ satisfy the integro-differential equation
\begin{equation} \label{Int-Deri}
  S'_{j0}(x,\mu) = g'_{j0}(x,\mu) + \int_0^x \partial_x G(x,t,\mu) q_{0}(t) S_{j0}(t,\mu) dt.
\end{equation}
Using the estimate (\ref{Green2}) on $\partial_x G(x,t,\mu)$, we get for all $x \in (0,A)$ and for all $\mu \in \C$ such that $\Re(\mu) \geq 0$, for $j = 1,2$,
\begin{equation} \label{Est-Sj0'-1}
  |S'_{j0}(x,\mu) - g'_{j0}(x,\mu)| \leq C \left( \frac{1+|\mu| x}{x} \right)^{\frac{1}{2}}  e^{ \Re(\mu) x} \left[ e^{C \left( \int_0^x \frac{t |q_0(t)|}{1+|\mu|t} dt \right)} - 1 \right], \quad j=1,2.
\end{equation}
We thus deduce from (\ref{Est-mu2}) that given a $x_0 \in (0,A)$ fixed, for all $x \in (0,x_0)$ and for $j=1,2$, we have as $|\mu| \to \infty$ with $\Re(\mu) \geq 0$
$$
 |S'_{j0}(x,\mu) - g'_{j0}(x,\mu)| \leq C |\mu|^\half \frac{e^{\Re(\mu) x}}{\log|\mu|^{\eo}}.
$$

We summarize all these results in a corollary.

\begin{coro} \label{UniEst-Sj0}
Let $x_0 \in (0,A)$ be fixed. Then for all $x \in (0,x_0)$ and for $j=1,2$, we have as $|\mu| \to \infty$ with $\Re(\mu) \geq 0$,
$$
 |S_{j0}(x,\mu)| \leq C \frac{e^{\Re(\mu) x}}{|\mu|^\half},
$$
$$
 |S_{j0}(x,\mu) - g_{j0}(x,\mu)| \leq C \frac{e^{\Re(\mu) x}}{|\mu|^\half \, \log|\mu|^{\eo}}.
$$
Moreover, we have
$$
 |S'_{j0}(x,\mu) - g'_{j0}(x,\mu)| \leq C |\mu|^\half \frac{e^{\Re(\mu) x}}{\log|\mu|^{\eo}}.
$$
\end{coro}

Symmetrically, we can construct the FSS $\{S_{11}, S_{21} \}$ and prove the corresponding estimates. We first define
\begin{eqnarray*}
  g_{11}(x,\mu) & = & C_{11} \Gamma(1 - i\lambda) \left( \frac{\mu}{2} \right)^{i\lambda} \left[ \sqrt{A-x} I_{i\lambda}(\mu (A-x)) + \frac{2i \sinh(\lambda \pi)}{\pi} \sqrt{(A-x)} K_{i\lambda}(\mu (A-x)) \right], \nonumber \\
	              & = & C_{11} \Gamma(1 - i\lambda) \left( \frac{\mu}{2} \right)^{i\lambda} \sqrt{A-x} I_{-i\lambda}(\mu (A-x)), \label{g11} \\
	g_{21}(x,\mu) & = & -\frac{1}{2i\lambda C_{11}} \Gamma(1 + i\lambda) \left( \frac{\mu}{2} \right)^{-i\lambda} \sqrt{A-x} I_{i\lambda}(\mu (A-x)) \label{g22}
\end{eqnarray*}
which are obtained using the hypothesis (\ref{AH-A}). Then using the same procedure as above, we obtain

\begin{coro} \label{UniEst-Sj1}
Let $x_0 \in (0,A)$. Then for all $x \in (x_0, A)$ and for $j=1,2$, as $|\mu| \to \infty$ with $\Re(\mu) \geq 0$
  $$
	  |S_{j1}(x,\mu)| \leq C \frac{e^{\Re(\mu) (A-x)}}{|\mu|^\half},
	$$
	$$
	  |S_{j1}(x,\mu) - g_{j1}(x,\mu)| \leq C \frac{e^{\Re(\mu) (A-x)}}{|\mu|^\half \, \log|\mu|^{\eu}}.
	$$
	$$
 |S'_{j1}(x,\mu) - g'_{j1}(x,\mu)| \leq C |\mu|^\half \frac{e^{\Re(\mu) (A-x)}}{\log|\mu|^{\eu}}.
$$
\end{coro}

We finish this Section showing the additional properties of the FSS $\{S_{10}, S_{20}\}$ and $\{S_{11}, S_{21}\}$. First, we can compute their Wronskians easily. Indeed, from (\ref{Est-Sj0-1}) and (\ref{Est-Sj0'-1}), it is clear that for $j=1,2$
$$
  S_{j0}(x,\mu) = g_{j0}(x,\mu) + o(x^\half), \quad \quad S'_{j0}(x,\mu) = g'_{j0}(x,\mu) + o(x^{-\half}), \quad x \to 0.
$$
Moreover, we deduce from (\ref{I-0}), (\ref{K-0}) and the corresponding estimates for the derivatives of the modified Bessel functions that
$$
  g_{j0}(x,\mu) = O(x^\half), \quad \quad g'_{j0}(x,\mu) = O(x^{-\half}), \quad x \to 0.
$$
Hence, we conclude with the help of (\ref{Goal}) that
\begin{equation}\label{Asympx-S10}
  S_{10}(x,\mu) \sim_{x \to 0} C_{10} x^{\frac{1}{2} - i\lambda}, \quad \quad S_{20}(x,\mu) \sim_{x \to 0} \frac{1}{2i\lambda C_{10}} x^{\frac{1}{2} + i\lambda},
\end{equation}
from which we can compute the Wronskian $W(S_{20}, S_{10})$ at $x = 0$. We find
$$
  W(S_{10}, S_{20}) = 1.
$$
Similarly we prove
$$
  W(S_{11}, S_{21}) = 1.
$$

Moreover it is clear from the general results on the analytic dependence of solutions of ODEs with respect to parameters that the functions $\mu \to S_{jm}(x,\mu)$ are entire and even for $j=1,2$, for $m=0,1$ and for $x$ fixed.

%We now show that the functions $\mu \to S_{jm}(x,\mu)$ are entire and even for $j=1,2$ and $m=0,1$ and $x$ fixed. From (\ref{defI}), it is clear that the function $\left( \frac{z}{2} \right)^{-\nu} I_\nu(z)$ is entire and even in $z$. but, from (\ref{defI}), (\ref{defK}), (\ref{g10}) and (\ref{g20}), we have
%\begin{eqnarray}
  %& g_{10}(x,\mu) = C_{10} \Gamma(1 - i\lambda) \left( \frac{\mu}{2} \right)^{i\lambda} \sqrt{x} I_{-i\lambda}(\mu x), \label{g10-Alt}\\
	%& g_{20}(x,\mu) = \frac{1}{2i\lambda C_{10}} \Gamma(1 + i\lambda) \left( \frac{\mu}{2} \right)^{-i\lambda} \sqrt{x} I_{i\lambda}(\mu x). \nonumber
%\end{eqnarray}
%Hence we conclude that the functions $g_{10}(x,\mu)$ and $g_{20}(x,\mu)$ are entire and even in $\mu$.
%
%Moreover, the Green Kernel (\ref{Green}) can be written as
%$$
  %G(x,t, \mu) = \sqrt{xt} \frac{\pi}{2 \sinh(\lambda \pi)} \left( I_{i\lambda}(\mu x) I_{-i\lambda}(\mu t) - I_{i\lambda}(\mu t) I_{-i\lambda}(\mu x) \right)
%$$
%and so is also entire and even in $\mu$. By induction, we deduce that each term $g_{j0}(x,\mu)$ of the iterated serie (\ref{IterSerie}) is entire and even in $\mu$. At last, using the uniform convergence of (\ref{IterSerie}), we prove the claim for $S_{j0}(x,\mu), \ j=1,2$. A similar proof shows the same result for $S_{j1}(x,\mu), \ j=1,2$.

\subsection{Estimates of $\Delta_q(\lambda)$ and $M_q(\lambda)$ for large angular momentum $\mu$}

In this Section, we use the results of the Corollaries \ref{UniEst-Sj0} and \ref{UniEst-Sj1} to obtain estimates of the characteristic function $\Delta_q(\mu^2)$ and the generalized Weyl-Titchmarsh function $M_q(\mu^2)$ for large angular momentum $\mu$. Let us start determining the asymptotics of the functions $S_{jm}(x,\mu)$ for $j=1,2$ and $m=0,1$ when $x$ is fixed in $(0,A)$ and $|\mu|$ is large.

We recall the following estimates of the modified Bessel function $I_\nu(z)$ as $|z|$ is large (see for instance \cite{Ol}).
\begin{equation} \label{Est-I-mu}
  I_\nu(z) = \frac{e^z}{\sqrt{2\pi z}} \left( 1 + O(\frac{1}{z}) \right) \pm \frac{e^{-z + i (\pm \nu \pi + \frac{\pi}{2}})}{\sqrt{2\pi z}} \left( 1 + O(\frac{1}{z}) \right),
\end{equation}
for $-\frac{\pi}{2} + \delta \leq \pm Arg(z) \leq \frac{\pi}{2}$ with $\delta >0$.

%Recalling for convenience reader (see  (\ref{g10-Alt}) and (\ref{g20})),
%\begin{eqnarray*}
%  & g_{10}(x,\mu) = C_{10} \Gamma(1 - i\lambda) \left( \frac{\mu}{2} \right)^{i\lambda} \sqrt{x} I_{-i\lambda}(\mu x), \\
%	& g_{20}(x,\mu) = \frac{1}{2i\lambda C_{10}} \Gamma(1 + i\lambda) \left( \frac{\mu}{2} \right)^{-i\lambda} \sqrt{x} I_{i\lambda}(\mu x),
% \end{eqnarray*}

We deduce from (\ref{g10}) and (\ref{g20}) that
\begin{equation} \label{Est-g10-mu}
  g_{10}(x,\mu) = C_{10} \frac{ \Gamma(1 - i\lambda)}{\sqrt{\pi} 2^{i\lambda + \half}} \mu^{i\lambda - \half} \left( e^{\mu x} [1] \pm e^{-\mu x \pm \lambda \pi + i \frac{\pi}{2}} [1] \right),
\end{equation}
\begin{equation} \label{Est-g20-mu}
  g_{20}(x,\mu) =  \frac{ \Gamma(1 + i\lambda)}{2i\lambda C_{10} \sqrt{\pi} 2^{-i\lambda + \half}} \mu^{-i\lambda - \half} \left( e^{\mu x} [1] \pm e^{-\mu x \mp \lambda \pi + i \frac{\pi}{2}} [1] \right),
\end{equation}
for $-\frac{\pi}{2} + \delta \leq \pm Arg(\mu) \leq \frac{\pi}{2}$ with $\delta >0$ and $x \in (0,A)$ fixed and where $[1] = 1 + O\left( \frac{1}{|\mu|} \right)$ as $|\mu| \to \infty$.

Combining the asymptotics (\ref{Est-g10-mu}) and (\ref{Est-g20-mu}) with Corollary \ref{UniEst-Sj0}, we finally get
\begin{equation} \label{Est-S10-mu}
  S_{10}(x,\mu) = C_{10} \frac{ \Gamma(1 - i\lambda)}{\sqrt{\pi} 2^{i\lambda + \half}} \mu^{i\lambda - \half} \left( e^{\mu x} [1]_0 \pm e^{-\mu x \pm \lambda \pi + i \frac{\pi}{2}} [1]_0 \right),
\end{equation}
\begin{equation} \label{Est-S20-mu}
  S_{20}(x,\mu) =  \frac{ \Gamma(1 + i\lambda)}{2i\lambda C_{10} \sqrt{\pi} 2^{-i\lambda + \half}} \mu^{-i\lambda - \half} \left( e^{\mu x} [1]_0 \pm e^{-\mu x \mp \lambda \pi + i \frac{\pi}{2}} [1]_0 \right),
\end{equation}
for $-\frac{\pi}{2} + \delta \leq \pm Arg(\mu) \leq \frac{\pi}{2}$ with $\delta >0$ and $x \in (0,A)$ fixed and where $[1]_0 = 1 + O\left( \frac{1}{\log|\mu|^{\eo}} \right)$ as $|\mu| \to \infty$.

Also, using the asymptotics for the derivatives of the Bessel functions $I_\nu(z)$ and $K_\nu(z)$ for large $z$, we get analogously
\begin{equation} \label{Est-S10'-mu}
  S'_{10}(x,\mu) = C_{10} \frac{ \Gamma(1 - i\lambda)}{\sqrt{\pi} 2^{i\lambda + \half}} \mu^{i\lambda + \half} \left( e^{\mu x} [1]_0 \mp e^{-\mu x \pm \lambda \pi + i \frac{\pi}{2}} [1]_0 \right),
\end{equation}
\begin{equation} \label{Est-S20'-mu}
  S'_{20}(x,\mu) =  \frac{ \Gamma(1 + i\lambda)}{2i\lambda C_{10} \sqrt{\pi} 2^{-i\lambda + \half}} \mu^{-i\lambda + \half} \left( e^{\mu x} [1]_0 \mp e^{-\mu x \mp \lambda \pi + i \frac{\pi}{2}} [1]_0 \right),
\end{equation}
for $-\frac{\pi}{2} + \delta \leq \pm Arg(\mu) \leq \frac{\pi}{2}$ with $\delta >0$ and $x \in (0,A)$ fixed.

At last, symmetrically, we can perform the same analysis with $S_{j1}, \ j=1,2$ and find
\begin{equation} \label{Est-S11-mu}
  S_{11}(x,\mu) = C_{11} \frac{ \Gamma(1 - i\lambda)}{\sqrt{\pi} 2^{i\lambda + \half}} \mu^{i\lambda - \half} \left( e^{\mu (A-x)} [1]_1 \pm e^{-\mu (A-x) \pm \lambda \pi + i \frac{\pi}{2}} [1]_1 \right),
\end{equation}
\begin{equation} \label{Est-S21-mu}
  S_{21}(x,\mu) =  -\frac{ \Gamma(1 + i\lambda)}{2i\lambda C_{11} \sqrt{\pi} 2^{-i\lambda + \half}} \mu^{-i\lambda - \half} \left( e^{\mu (A-x)} [1]_1 \pm e^{-\mu (A-x) \mp \lambda \pi + i \frac{\pi}{2}} [1]_1 \right),
\end{equation}
\begin{equation} \label{Est-S11'-mu}
  S'_{11}(x,\mu) = C_{11} \frac{ \Gamma(1 - i\lambda)}{\sqrt{\pi} 2^{i\lambda + \half}} \mu^{i\lambda + \half} \left( -e^{\mu (A-x)} [1]_1 \pm e^{-\mu (A-x) \pm \lambda \pi + i \frac{\pi}{2}} [1]_1 \right),
\end{equation}
\begin{equation} \label{Est-S21'-mu}
  S'_{21}(x,\mu) =  -\frac{ \Gamma(1 + i\lambda)}{2i\lambda C_{11} \sqrt{\pi} 2^{-i\lambda + \half}} \mu^{-i\lambda + \half} \left( - e^{\mu (A-x)} [1]_1 \pm e^{-\mu (A-x) \mp \lambda \pi + i \frac{\pi}{2}} [1]_1 \right),
\end{equation}
for $-\frac{\pi}{2} + \delta \leq \pm Arg(\mu) \leq \frac{\pi}{2}$ with $\delta >0$ and $x \in (0,A)$ fixed and where $[1]_1 = 1 + O\left( \frac{1}{\log|\mu|^{\eu}} \right)$ as $|\mu| \to \infty$.

We can now compute the asymptotics of the characteristic and generalized Weyl-Titchmarsh functions $\Delta_q(\mu^2)$ and $M_q(\mu^2)$ from the previous ones. From (\ref{Char}), (\ref{WT}) and (\ref{Est-S10-mu}) - (\ref{Est-S21'-mu}), we obtain

\begin{prop} \label{Asymp-WT}
The following asymptotics hold for $-\frac{\pi}{2} + \delta \leq \pm Arg(\mu) \leq \frac{\pi}{2}$ with $\delta >0$ as $|\mu| \to \infty$
  $$
	  \Delta_q(\mu^2) = \frac{C_{10}C_{11} \Gamma(1-i\lambda)^2}{\pi 2^{2 i \lambda}} \mu^{2i \lambda} e^{\pm \lambda \pi} 2 \cosh(\mu A \mp \lambda \pi) [1]_2,
	$$
	$$
	  \delta_q(\mu^2) = W(S_{11}, S_{20}) = \frac{C_{11} \Gamma(1-i\lambda) \Gamma(1+i\lambda)}{C_{10} 2 i \lambda \pi} 2 \cosh(\mu A) [1]_2,
	$$
	$$
	  M_q(\mu^2) = -\frac{\delta_q(\mu)}{\Delta_q(\mu)} = - \frac{\Gamma(1+i\lambda) e^{\mp \lambda \pi} 2^{2 i \lambda}}{2 i \lambda C_{10}^2 \Gamma(1-i\lambda)} \mu^{-2i \lambda} \frac{\cosh(\mu A)}{\cosh(\mu A \mp \lambda \pi)} [1]_2,
	$$
	where $[1]_2 = 1 + O\left( \frac{1}{\log|\mu|^{\epsilon}} \right)$ as $|\mu| \to \infty$ and $\epsilon:= min\ (\eo, \eu)$.
\end{prop}

We shall need the following Corollary of Proposition \ref{Asymp-WT} in the next Section.

\begin{coro} \label{Bounded-iR}
  The functions $ \mu \mapsto \Delta_q(\mu^2)$ and $\mu \mapsto \delta_q(\mu^2)$ are of exponential type and are bounded on the imaginary axis $i\R$. More precisely, we have
	$$
	  \forall \mu\in \C, \quad |\Delta_q(\mu)|, \ |\delta_q(\mu)| \leq C^{|\Re(\mu)|A},
	$$
	and
	$$
	  \forall y \in \R, \quad |\Delta_q(iy)|, \ |\delta_q(iy)| \leq C.
	$$
\end{coro}

\begin{rems}
  1) Putting $\rho = -i\mu$ in Proposition \ref{Asymp-WT}, we obtain exactly the asymptotics of \cite{FY} up to a multiplicative constant (encoded by the degree of freedom in the choice of $C_{10}$ and $C_{11}$). The only difference is that our remainder is in $O\left( \frac{1}{\log|\mu|^{\epsilon}} \right)$ as $|\mu| \to \infty$ whereas their remainder was in $0 \left( \frac{1}{|\mu|^{\beta}} \right)$ for a given $\beta > 0$ as $|\mu| \to \infty$.  \\
	
\noindent 2) For later use, we need to know the asymptotics of $\dot{\Delta}_q(\mu^2) : = \frac{d \Delta_q(\mu^2)}{d\mu^2}$ as $\mu \to \infty$ and $\mu \in \R$. Since it has been done in \cite{FY}, eq. (3.15), we only give the result here without proof.
\begin{equation} \label{Delta'}
	 \dot{\Delta}_q(\mu^2) = \frac{C_{10}C_{11} \Gamma(1-i\lambda)^2 A}{\pi 2^{2 i \lambda}} \mu^{2i \lambda} e^{\pm \lambda \pi} 2 \sinh(\mu A \mp \lambda \pi) [1]_2.
\end{equation}	
As a consequence, we have

\begin{prop} \label{Increasing}
  The function $\mu^2 \in \R^+ \mapsto |\Delta_q(\mu^2)|$ is strictly increasing for large enough $\mu >>1$.
\end{prop}	
\begin{proof}
  Writing $\frac{d |\Delta_q(\mu^2)|^2}{d\mu^2} = 2 Re \left( \overline{\Delta_q(\mu^2)} \dot{\Delta}_q(\mu^2) \right)$, we get from Proposition \ref{Asymp-WT}  and (\ref{Delta'})
	$$
	  \frac{d |\Delta_q(\mu^2)|^2}{d\mu^2} = K \sinh(2 \mu A \mp 2 \lambda \pi) [1]_2,
	$$
	for a positive constant $K$. Hence $\frac{d |\Delta_q(\mu^2)|^2}{d\mu^2} > 0$ for large enough (real) $\mu >>1$ which proves the assertion.
\end{proof}

\noindent 3) As shown in \cite{FY}, the asymptotics of $\Delta_q(\mu^2)$ in Proposition \ref{Asymp-WT} as well as a standard Rouch\'e argument allow us to determine precise asymptotics for the zeros $(\alpha_n)_{n \in \Z}$ of $\Delta_q(\mu^2)$. Precisely, we obtain
	\begin{equation} \label{ReggePoles-1}
	  \alpha_n = \frac{\lambda \pi}{A} + i\frac{(n+\half + p) \pi}{A} + O \left( \frac{1}{(\log n)^{\epsilon}} \right), \quad Arg(\alpha_n) \in [0, \pi], \quad n \in \N, \ n \to\infty,
	\end{equation}
	and by symmetry
	\begin{equation} \label{ReggePoles-2}
	  \alpha_{-n} = -\frac{\lambda \pi}{A} - i\frac{(n+\half + p) \pi}{A} + O\left( \frac{1}{(\log n)^{\epsilon}} \right), \quad Arg(\alpha_n) \in [-\pi, 0], \quad n \in \N, \ n \to\infty,
	\end{equation}
	for a certain $p \in \Z$. Similarly, we obtain for the zeros $(\beta_n)_{n \in \Z}$ of $\delta_q(\mu^2)$
	\begin{equation} \label{ReggePoles-3}
	  \beta_n = i\frac{(n+\half + p) \pi}{A} + O \left( \frac{1}{(\log n)^{\epsilon}} \right), \quad Arg(\beta_n) \in [0, \pi], \quad n \in \N, \ n \to \infty,
	\end{equation}
	\begin{equation} \label{ReggePoles-4}
	  \beta_{-n} = -i\frac{(n+\half + p) \pi}{A} + O \left( \frac{1}{(\log n)^{\epsilon}} \right), \quad Arg(\beta_n) \in [-\pi, 0], \quad n \in \N, \ |n| \to\infty,
	\end{equation}
	for a certain $p \in \Z$. \\
	
	\noindent 4) Finally , fix $\delta > 0$ and denote $G_\delta = \{ \mu \in \C, \ |\mu - \alpha_n| \geq 0, n \in \Z \}$. Then it proved in \cite{FY}, eq. (3.16) the following estimate from below
	\begin{equation} \label{Delta-Minoration}
	  |\Delta_q(\mu^2)| \geq C \  e^{|\Re(\mu)| A}, \quad \forall \mu \in G_\delta.
	\end{equation}
\end{rems}

\subsection{Scattering operator Vs. Characteristic and Weyl-Titchmarsh operators}

In this Section, we prove Proposition \ref{WT-Reflection} and give a few consequences.

First, observe that the scattering operator defined in Theorem \ref{Stat-Sol} leaves invariant the span of each generalized spherical harmonic $Y_{n}$. Hence, it suffices to calculate the scattering operator on each vector space generated by the $Y_{n}$'s.

To do this, we recall from Theorem \ref{Stat-Sol} that, given any solution $f = u_n(x) Y_n(y) \in \B^*$ of (\ref{ShiftedWaveEq}), there exists a unique $\psi_n^{(\pm)} = (\psi_{0n}^{(\pm)}, \psi_{1n}^{(\pm)}) \in \C^2$ such that
\begin{eqnarray} \label{IKL-Asymp}
  	u_n(x) \simeq & \omega_-(\lambda) \left( \chi_0(x) \,x^{\half + i\lambda} \psi_{0n}^{(-)} + \chi_1(x) \,(A-x)^{\half + i\lambda} \psi_{1n}^{(-)} \right) \\
		         & - \omega_+(\lambda) \left( \chi_0(x) \, x^{\half - i\lambda} \psi_{0n}^{(+)} + \chi_1(x) \, (A-x)^{\half - i\lambda} \psi_{1n}^{(+)} \right), \nonumber
\end{eqnarray}
where $\omega_\pm(\lambda)$ is given by (\ref{Omega}) and the cutoffs $\chi_{0}$ and $\chi_1$ are defined in (\ref{Cutoff}). We apply this result to the FSS $\{S_{jm}, \ j=1,2,\ m=0,1\}$ obtained in the previous Section.

From now on, we fix a given $n \in \N$ and omit the indices $n$ in the next calculations. We compute first the $\psi_j^{(\pm)}, \ j=0,1$ corresponding to the solutions $S_{10}$. From (\ref{Asympx-S10}), we know that
$$
  S_{10}(x,\mu) \sim_{x \to 0} \chi_0(x) C_{10} x^{\half - i\lambda}.
$$
Comparing with (\ref{IKL-Asymp}), we obtain for $S_{10}$ the following associated vector
\begin{equation} \label{0pm}
  \psi_0^{(-)} = 0, \quad \psi_0^{(+)} = -\frac{C_{10}}{\omega_+(\lambda)}.
\end{equation}

Let us write now $S_{10}$ as a linear combination of $S_{11}, S_{21}$, \textit{i.e.}
$$
  S_{10} = a_1(\mu) S_{11} + b_1(\mu) S_{21}.
$$
Since $W(S_{11},S_{21}) = 1$, we get immediately
\begin{equation} \label{a1b1}
  a_1(\mu) = W(S_{10}, S_{21}), \quad b_1(\mu) = W(S_{11}, S_{10}).
\end{equation}
But recall that (similarly to (\ref{Asympx-S10}))
$$
  S_{11}(x,\mu) \sim_{x \to A} \chi_1(x) C_{11} (A-x)^{\half - i\lambda}, \quad S_{21}(x,\mu) \sim_{x \to A} \chi_1(x) \frac{-1}{2i\lambda C_{11}} (A-x)^{\half + i\lambda}.
$$
Hence, comparing again with (\ref{IKL-Asymp}), we obtain for $S_{10}$ the following associated vector
\begin{equation} \label{1pm}
  \psi_1^{(-)} = \frac{-1}{2i\lambda C_{11} \omega_-(\lambda)} b_1(\mu) , \quad \psi_1^{(+)} = -\frac{C_{11}}{\omega_+(\lambda)} a_1(\mu).
\end{equation}
Summarizing, we find that the solution $S_{10}$ satisfies the decomposition (\ref{IKL-Asymp}) and is associated to the following vectors
\begin{equation} \label{S10-a1B1}
  \psi^{(-)} = \left( \begin{array}{c} 0 \\ \frac{-1}{2i\lambda C_{11} \omega_-(\lambda)} b_1(\mu) \end{array} \right), \quad \psi^{(+)} = \left( \begin{array}{c} -\frac{C_{10}}{\omega_+(\lambda)} \\ -\frac{C_{11}}{\omega_+(\lambda)} a_1(\mu) \end{array} \right).
\end{equation}

Using a similar procedure, we find that the solution $S_{11}$ is associated to the following vectors in the decomposition (\ref{IKL-Asymp})
\begin{equation} \label{S11-a0B0}
  \phi^{(-)} = \left( \begin{array}{c} \frac{1}{2i\lambda C_{10} \omega_-(\lambda)} b_0(\mu) \\ 0 \end{array} \right), \quad \phi^{(+)} = \left( \begin{array}{c} -\frac{C_{10}}{\omega_+(\lambda)} a_0(\mu) \\ -\frac{C_{11}}{\omega_+(\lambda)} \end{array} \right),
\end{equation}
where
\begin{equation} \label{a0b0}
  a_0(\mu) = W(S_{11}, S_{20}), \quad b_0(\mu) = W(S_{10}, S_{11}).
\end{equation}

Recall now that for any $\psi_n^{(-)} \in \C_2$, there exists a unique $\psi_n^{(+)} \in \C_2$ and $u_n(x) Y_n \in \B^*$ satisfying (\ref{ShiftedWaveEq}) for which the expansion (\ref{IKL-Asymp}) above holds. This defines the scattering operator $S(\lambda, \mu_n)$ as the $2 \times 2$ matrix such that for all $\psi_n^{(-)} \in \C^2$
\begin{equation} \label{Sn}
  \psi_n^{(+)} = S(\lambda, \mu_n) \psi_n^{(-)}.
\end{equation}
Using the notation
$$
  S(\lambda,\mu_n) = \left[ \begin{array}{cc} L(\lambda,\mu_n) & T_L(\lambda,\mu_n) \\ T_R(\lambda,\mu_n) & R(\lambda,\mu_n) \end{array} \right],
$$
and using the definition (\ref{Sn}) of the partial scattering matrix together with (\ref{S10-a1B1}) - (\ref{S11-a0B0}), we find
\begin{equation} \label{TnLn}
  S(\lambda,\mu_n) = \left[ \begin{array}{cc} -\frac{2i\lambda C_{10}^2 \omega_-(\lambda)}{\omega_+(\lambda)} \frac{a_0(\mu_n)}{b_0(\mu_n)} & \frac{2i\lambda C_{10} C_{11} \omega_-(\lambda)}{\omega_+(\lambda)} \frac{1}{b_1(\mu_n)} \\ -\frac{2i\lambda C_{10} C_{11} \omega_-(\lambda)}{\omega_+(\lambda)} \frac{1}{b_0(\mu_n)} &  \frac{2i\lambda C_{11}^2 \omega_-(\lambda)}{\omega_+(\lambda)} \frac{a_1(\mu_n)}{b_1(\mu_n)} \end{array} \right].
\end{equation}
In this expression of the partial scattering matrix, we recognize the usual transmission coefficients $T_{L}(\lambda,\mu_n), \, T_{R}(\lambda,\mu_n)$ and the reflection coefficients $L(\lambda,\mu_n)$ (from the left) and $R(\lambda,\mu_n)$ (from the right). Since they are written in terms of Wronskians of the $S_{jm}, \, j=1,2, \ m=0,1$, we can make the link between the characteristic function (\ref{Char}) and generalized Weyl-Titchmarsh function (\ref{WT}) as follows. Remarking that
\begin{equation} \label{Link}
  \Delta_q(\mu^2) = b_1(\lambda,\mu) = - b_0(\lambda,\mu), \quad M_q(\mu^2) = \frac{a_0(\mu)}{b_0(\mu)},
\end{equation}
we eventually get
\begin{eqnarray}
  L(\lambda,\mu_n) = -\frac{2i\lambda C_{10}^2 \omega_-(\lambda)}{\omega_+(\lambda)} M_q(\mu_n^2), \label{L-WT} \\
	T(\lambda,\mu_n) = T_L(\lambda,\mu_n) = T_R(\lambda,\mu_n) = \frac{2i\lambda C_{10} C_{11} \omega_-(\lambda)}{\omega_+(\lambda)} \frac{1}{\Delta_q(\mu_n^2)} \label{T-Char}.
\end{eqnarray}
Hence, we have proved that the formulae stated in Theorem \ref{WT-Reflection} are indeed valid for each harmonic $Y_{n}$. The full theorem is thus proved summing up onto the whole harmonics.

Let us finish this Section with a few remarks concerning the scattering operator $S(\lambda)$. The scattering operator $S(\lambda)$ is unitary on $\H_\infty$ or equivalently, the matrices $S(\lambda,\mu_n)$ are unitary for all $n \in \N$. In fact, we can prove that for all $\mu \in \R$, we have
\begin{eqnarray}
  |T(\lambda,\mu)|^2 + |L(\lambda,\mu)|^2 = 1, \\
	|T(\lambda,\mu)|^2 + |R(\lambda,\mu)|^2 = 1, \\
	L(\lambda,\mu) \bar{T}(\lambda,\mu) + T(\lambda,\mu) \bar{R}(\lambda,\mu) = 0. \label{unitary}
\end{eqnarray}

So, using (\ref{L-WT}), (\ref{T-Char}) and (\ref{unitary}), we obtain:
\begin{equation}\label{expression1R}
R(\lambda, \mu_n) =  \frac{2i\lambda   \mid C_{10}\mid^2 C_{11} \omega_-(\lambda)}{\overline{C_{11}} \omega_+(\lambda)} \
\frac{\overline{\Delta_q}(\mu_n^2)}{\Delta_q (\mu_n^2)} \ \overline{M_q( \mu_n^2)}.
\end{equation}
On the other hand, by (\ref{TnLn}), we have
\begin{equation}\label{expression2R}
R(\lambda, \mu_n) =  \frac{2i\lambda   C_{11}^2 \omega_-(\lambda)}{\omega_+(\lambda)} \
\frac{a_1(\lambda,\mu_n)}{\Delta_q(\mu_n^2)}.
\end{equation}
We deduce from (\ref{expression1R}) and(\ref{expression2R}) the following expression for the reflexion coefficient $R(\lambda, \mu_n)$:
\begin{equation} \label{expressionR}
R(\lambda, \mu_n) =  \frac{2i\lambda   C_{11}^2 \mid C_{10}\mid^2 \omega_-(\lambda)}{\mid C_{11}\mid^2 \omega_+(\lambda)} \ M_q(\mu_n^2).
\end{equation}

\noindent
Using (\ref{L-WT}) - (\ref{T-Char}),  and the fact that $|\omega_-(\lambda)|^2 = |\omega_+(\lambda)|^2$ , we thus obtain for all $\mu \in \R$,
\begin{equation} \label{UnitaryRelation}
  4\lambda^2 |C_{10}|^2 \left[ |M_q(\mu^2)|^2 |C_{10}|^2 + \frac{1}{|\Delta_q(\mu^2)|^2}|C_{11}|^2 \right] = 1,
\end{equation}
or equivalently
\begin{equation} \label{Unitary-Delta}
  \frac{1}{4\lambda^2 |C_{10}|^2 |C_{11}|^2 } |\Delta_q(\mu^2)|^2 - \frac{|C_{10}|^2}{|C_{11}|^2} |\delta_q(\mu^2)|^2 = 1.
\end{equation}
Recalling that $\mu \mapsto \Delta_q(\mu^2), \delta_q(\mu^2)$ are entire, we deduce from (\ref{Unitary-Delta}) that
\begin{equation} \label{Unitary-DeltaExt}
  \frac{1}{4\lambda^2 |C_{10}|^2 |C_{11}|^2 } \Delta_q(\mu^2) \overline{\Delta_q(\bar{\mu}^2)} - \frac{|C_{10}|^2}{|C_{11}|^2} \delta_q(\mu^2) \overline{\delta_q(\bar{\mu}^2)} = 1, \quad \forall \mu \in \C.
\end{equation}

We shall use the following result in the next Section.

\begin{coro} \label{Increasing-Char-WT}
  For $\lambda \not=0$ and for all $\mu \in \R$ we have \\
	1) $|\Delta_q(\mu^2)| \geq 2 |\lambda| |C_{10}| |C_{11}|$, \\
	2) $|\Delta_q(\mu^2)|$ is strictly increasing to $+\infty$ for large enough $\mu$ as $\mu \to \infty, \ \mu \in \R$. \\
	3) $|M_q(\mu^2)| \leq \frac{1}{2|\lambda| |C_{10}|^2}$, \\
	4) $|M_q(\mu^2)|$ is strictly increasing to $\frac{1}{2|\lambda| |C_{10}|^2}$ for large enough $\mu$ as $\mu \to \infty, \ \mu \in \R$.
\end{coro}
\begin{proof}
The points 1) and 2) follow immediately from (\ref{UnitaryRelation}) and from Propositions \ref{Asymp-WT} and \ref{Increasing}. The points 3) and 4) are then consequence of 1), 2) and (\ref{UnitaryRelation}).
\end{proof}

%%%%%%%%%%%%%%%%%%%%%%%%%%%%%%%%%%%%%%%%%%%%%% THE INVERSE PROBLEM %%%%%%%%%%%%%%%%%%%%%%%%%%%%%%%%%%%%%%%%%%%%%%%%%%%%%%%%%%%5

\Section{The inverse problem at fixed energy} \label{IP}

In this Section , we prove Theorem \ref{Main}, that is we consider two asymptotically hyperbolic Liouville surfaces $(\S,g)$ and $(\tilde{\S}, \tilde{g})$ described in Definition \ref{AHLS}. Recall that we shall add a $\ \tilde{} \ $ to any quantities related to $(\tilde{\S}, \tilde{g})$. Our main assumptions are $B = \tilde{B}$ and given $\lambda \not=0$ be a fixed energy of the stationary wave equation (\ref{ShiftedWaveEq}), assume that the generalized Weyl-Titchmarsh operators coincide
$$
  M(\lambda) = \tilde{M}(\lambda),
$$
as operators from $\H_B$ to $\H_B$.

\subsection{The inverse problem for the angular equation}

By construction and hypothesis, we know that the generalized Weyl-Titchmarsh operators $M(\lambda)$ and $\tilde{M}(\lambda)$ both act on $\H_B$\footnote{Note that this is here where we use our assumption $B = \tilde{B}$ in order to compare objects acting on the same Hilbert space.} and are diagonalizable on the Hilbert bases $\{Y_{n\lambda}\}$ and $\{\tilde{Y}_{n\lambda}\}$ of $\H_B$ associated to the eigenvalues $M_q(\mu_{n\lambda}^2)$ and $M_{\tilde{q}}(\tilde{\mu}_{n\lambda}^2)$ respectively. Recall that $\mu_{n\lambda}^2$ are the eigenvalues of the angular separated second order ODE (\ref{AngularODE}). Thus, the multiplicity of $\mu_{n\lambda}^2$ is at most two. In what follows, we work at a fixed non-zero energy $\lambda$ and we drop the subscript $_\lambda$ from the notation from now on.

In this Section, we prove the following result:

\begin{prop} \label{AngularInvPb}
Assume that for a non-zero fixed energy $\lambda$, either $\Delta(\lambda) = \tilde{\Delta}(\lambda)$ or  $M(\lambda) = \tilde{M}(\lambda)$.
%Then there exists $N_1 \in \N$ such that for all $n \geq N_1$,
%$$
%\tilde{Y}_n = c_n Y_n, \quad |c_n| = 1,
%$$
%where the $Y_n$ and $\tilde{Y_n}$'s satisfy the angular equations (\ref{AngularODE})
%$$
%  -Y_n'' + (\lambda^2 +\frac{1}{4})\ b(y) Y_n = -\mu_n^2 Y_n, \quad -\tilde{Y}_n'' + (\lambda^2 +\frac{1}{4})\ \tilde{b}(y) \tilde{Y}_n = -\mu_n^2 \tilde{Y}_n.
%$$
Then there exists a constant $C$ such that
\begin{equation}\label{Inv-b}
  b(y) = \tilde{b}(y) + C, \quad \forall y \in (0,B),
\end{equation}
and
\begin{equation} \label{Inv-mu}
  \mu_n^2 = \tilde{\mu}_n^2 +  C (\lambda^2 +\frac{1}{4}), \quad \forall n  \in \N.
\end{equation}
\end{prop}

\begin{proof}
We give the proof for $M(\lambda)$ since the proof for $\Delta(\lambda)$ is similar. By construction and by our main hypothesis, we know that
$$
 \{ M_q(\mu_n^2) \}_{n \in \N} = \{ M_{\tilde{q}}(\tilde{\mu}_n^2) \}_{n \in \N}.
$$
For a fixed $n \in \N$, let us denote $E(\mu_{n}^2) = Span \ \{Y_m, \ M_q(\mu_m^2) = M_q(\mu_n^2) \}$ the eigenspace associated to $M_q(\mu_n^2)$. We also obtain from our main hypothesis that
$$
  \{ E(\mu_{n}^2) \}_{n \in \N} = \{ \tilde{E}(\tilde{\mu}_{n}^2) \}_{n \in \N}.
$$	

Since $\mu \rightarrow |M_q (\mu^2)|$ is {\it{strictly}} increasing to infinity for large enough real $\mu$ and since $\mu_{n}^2 \to \infty$, there exists $N \geq 0$ such that for
$n \geq N$, $dim\ (E(\mu_n^2)) =1 \ {\rm{or}} \  2$. For some $n \geq N$ fixed, let us assume for instance that $dim\ (E(\mu_n^2)) =1$, (the proof in the case where $dim\ (E(\mu_n^2)) =2$ is quite similar using an elementary linearity argument). 

Thus, there exists $p \in \N$ such that $E(\mu_{n}^2) = \tilde{E}(\tilde{\mu}_{p}^2)$. From the equations satisfied by the $Y_n$ and $\tilde{Y_n}$'s and setting $\tilde{Y_p} = c Y_n$ with $|c|=1$, we get easily
$$
	  \left( (\lambda^2+ \frac{1}{4}) \ b(y) - \mu_n^2 \right) Y_n(y) = \left( (\lambda^2+ \frac{1}{4}) \ \tilde{b}(y) - \tilde{\mu}_p^2 \right) Y_n(y), \quad \forall y \in (0,B).
$$
But it is known that the solutions $Y_n$ of the angular separated equation (\ref{AngularODE}) has only a finite number of zeros in (0,\,B) (see \cite{Ze}, Theorem 4.3.1). Hence away from this finite set of zeros $\{y_1, \dots, y_k\}$, we obtain
	\begin{equation} \label{pp5}
	  (\lambda^2+ \frac{1}{4})\  b(y) - \mu_n^2 = (\lambda^2+ \frac{1}{4}) \ \tilde{b}(y) - \tilde{\mu}_p^2, \quad \forall y \in (0,B) \setminus \{y_1, \dots, y_k\}.
	\end{equation}
	By continuity of $b, \tilde{b}$, we finally get from (\ref{pp5})
	\begin{equation} \label{pp6}
	  b(y) - \tilde{b}(y) = \frac{\mu_n^2 - \tilde{\mu}_p^2}{\lambda^2+ \frac{1}{4}}, \quad \forall y \in (0,B).
	\end{equation}
Since the left-hand-side only depends on the angular variable $y$ and the right-hand-side is a constant, we thus conclude that there exists a constant $C \in \R$ such that (\ref{Inv-b}) holds.

Recall that the angular separated ODE (\ref{AngularODE}) only depends on $B$ and the function $b$. Thus the eigenvalues $\mu_n^2$ of the equation also only depend on $B$ and the function $b$. We deduce then from (\ref{Inv-b}), (\ref{pp6}) and the ordering of the $\mu_{n}^2$ that
$$
  \mu_n^2 = \tilde{\mu}_n^2 +  C (\lambda^2+ \frac{1}{4}), \quad \forall n \in \N.
$$
\end{proof}

\begin{rem} \label{RE}
  In Proposition \ref{AngularInvPb}, we first proved that the angular separated ODEs for the two asymptotically hyperbolic Liouville surfaces $(\S,g)$ and $(\tilde{\S}, \tilde{g})$ coincide (up to a constant) from our main hypothesis. We deduced from this that the eigenvalues $\mu_n^2$  and $\tilde{\mu}_n^2$ of these angular ODEs coincide (up to a constant). Note that this does not imply that the corresponding eigenfunctions coincide since the eigenvalues can have multiplicity $2$. But, since the definition of the generalized WT operator is independent of the choice of the eigenfunctions $Y_{n}$, we can assume in the following that for all $n \in \N$, $Y_n = \tilde{Y}_n$ and thus, $M_q(\mu_n^2) = M_{\tilde{q}}(\tilde{\mu}_n^2)$.
\end{rem}

\subsection{Further reduction of the problem}

Recall that the generalized Weyl-Titchmarsh function $M_q(\mu^2)$ is entire and \emph{even} in $\mu$. Thus from Proposition \ref{AngularInvPb} and Remark \ref{RE}, we deduce from our main assumptions that there exists a constant $C$ such that
$$
  \forall n \in \N, \quad M_q(\mu_n^2) = M_{\tilde{q}}(\tilde{\mu}_n^2) =  M_{\tilde{q}}(\mu_n^2 - C (\lambda^2 + \frac{1}{4})),
$$
and also
$$
  \forall n \in \N, \quad \Delta_q(\mu_n^2) = \Delta_{\tilde{q}}(\tilde{\mu}_n^2) =  \Delta_{\tilde{q}}(\mu_n^2 - C (\lambda^2 + \frac{1}{4})),
$$
We first slightly modify the above equalities using the following Lemma.

\begin{lemma} \label{Sym}
  $$
	  \forall L \in \C, \ \forall \mu \in \C, \quad \Delta_q(\mu^2) = \Delta_{q + L}(\mu^2 - L), \quad M_q(\mu^2) = M_{q + L}(\mu^2 - L).
	$$
\end{lemma}

\begin{proof}
Recall that the FSS given by $S_{jm}(x,\mu), \ j=1,2, \ m= 0,1$ are solutions of (\ref{RadialODE})
$$
	-u'' + q(x) u = - \mu^2 u, \quad \forall \mu \in \C,
$$
satisfying the asymptotics (\ref{Jost0}) and (\ref{JostA}). The characteristic and generalized WT functions are then expressed in terms of Wronskians of the $S_{jm}(x,\mu), \ j=1,2, \ m= 0,1$ by (\ref{Char}) and (\ref{WT}).
	
But note that, for all $L \in \C$, the FSS $S_{jm}(x,\mu), \ j=1,2, \ m= 0,1$ are also solutions of
\begin{equation} \label{RadODE-Mod}
	-u'' + (q(x) + L) u = - (\mu^2 - L) u, \quad \forall \mu \in \C.
\end{equation}
The asymptotics of the new potential $q(x) + L$ at $x=0$ and $x=B$ do not change when we add a constant $L$ and are still given by (\ref{Asymp}). Hence we can define the characteristic and generalized WT functions for (\ref{RadODE-Mod}) using the same FSS $S_{jm}(x,\mu), \ j=1,2, \ m= 0,1$. The characteristic and generalized WT functions for (\ref{RadODE-Mod}) thus coincide with the old ones. We deduce from this fact the transformation laws given in the Lemma.
\end{proof}

Using Lemma \ref{Sym}, we can assume from now on that
\begin{equation} \label{Assump2}
  \forall n \in \N, \quad M_q(\mu_n^2) = M_{\tilde{q} - C (\lambda^2+\frac{1}{4}) }(\mu_n^2),
\end{equation}
and
\begin{equation} \label{Assump3}
  \forall n \in \N, \quad \Delta_q(\mu_n^2) = \Delta_{\tilde{q} - C (\lambda^2+\frac{1}{4}) }(\mu_n^2).
\end{equation}

\subsection{The CAM method}

We shall use now the Complex Angular Momentum method (CAM) to prove that the equality (\ref{Assump2}) extends meromorphically for $\mu \in \C$ (since $M_q(\mu^2)$ has poles at the $(\alpha_n)_{n \in \Z}$). Recall that the functions
$$
  \Delta_q(\mu^2) = W(S_{11}, S_{10}), \quad \delta_q(\mu^2) = W(S_{11}, S_{20}),
$$
are entire functions in the variable $\mu$ that satisfy the estimates in Corollary \ref{Bounded-iR}. We can use these estimates to prove that the functions $\mu \mapsto \Delta_q(\mu^2), \delta_q(\mu^2)$ belong to the Cartwright class (\cite{Lev}) defined by

\begin{defi} \label{Cartwright}
  We say that an entire function $f$ belongs to the Cartwright class $\mathcal{C}$ if $f$ is of exponential type (\textit{i.e.} $|f(\mu)| \leq C e^{A|\mu|}$ for some positive constants $A, C$) and satisfies
	$$
	  \int_{\R} \frac{\log^+(|f(iy)|)}{1+y^2} dy < \infty,
	$$
	where
	$$
	  \log^+(x) = \left\{ \begin{array}{cc} \log(x) & \log(x) \geq 0, \\ 0, & \log(x) < 0. \end{array} \right.
	$$
\end{defi}

\begin{rem} \label{Nevanlinna}
   Note that if we restrict the Cartwright class $\mathcal{C}$ to functions analytic on the half plane $\C^+ = \{ \mu \in \C, \ \Re(\mu) \geq 0 \}$, then it reduces to the Nevanlinna class $N^+(\C^+)$ used in \cite{DN3, DGN} (see \cite{Lev}, Remark, p 116).
\end{rem}

It is well known that the zeros of entire functions in the Cartwright class $\mathcal{C}$ have a certain distribution in the complex plane (see \cite{Lev}, Theorem 1, p 127). As a consequence, the following uniqueness Theorem holds.

\begin{prop} \label{UniquenessC}
  Let $f \in \mathcal{C}$ and $(z_n)_n$ be the zeros of $f$. If there exists a subset $\mathcal{L} \subset \N$ such that
	$$
	  \sum_{n \in \mathcal{L}} \frac{\Re(z_n)}{|z_n|^2} = \infty,
	$$
	then $f = 0$ on $\C$.
\end{prop}

We can apply these results to the entire functions $\mu \mapsto \Delta_q(\mu^2), \delta_q(\mu^2)$.

\begin{coro} \label{UniquenessDelta}
  1) The functions $\Delta_q$ and $\delta_q$ belong to $\mathcal{C}$. \\
	2) Let $\mathcal{L} \subset \N$ such that $\ds \sum_{n \in \mathcal{L}} \frac{1}{n} = \infty$. Assume that one of the following equalities hold
	\begin{eqnarray*}
	  \Delta_q(\mu_n^2) = \tilde{\Delta}_{\tilde{q} - C(\lambda^2+\frac{1}{4})}(\mu_n^2), \quad \forall n \in \mathcal{L}, \\
		\delta_q(\mu_n^2) = \tilde{\delta}_{\tilde{q} - C(\lambda^2+\frac{1}{4})}(\mu_n^2), \quad \forall n \in \mathcal{L}.
	\end{eqnarray*}
	Then the corresponding equalities hold.
	\begin{eqnarray*}
	  \Delta_q(\mu^2) = \tilde{\Delta}_{\tilde{q} - C(\lambda^2+\frac{1}{4})}(\mu^2), \quad \forall \mu \in \C, \\
		\delta_q(\mu^2) = \tilde{\delta}_{\tilde{q} - C(\lambda^2+\frac{1}{4})}(\mu^2), \quad \forall \mu \in \C.
	\end{eqnarray*}
\end{coro}

\begin{proof}
  The first point is an immediate consequence of Definition \ref{Cartwright} and Corollary \ref{Bounded-iR}. The second point follows then from Proposition \ref{UniquenessC} and the M\"untz condition (\ref{MuntzCond}) satisfied by the $\mu_n$'s.
\end{proof}

Lastly, we need a slight improvement of the previous result that takes into account the fact that we know the coefficients $M_q(\mu^2)$, and not $\Delta_q(\mu^2)$ and $\delta_q(\mu^2)$ separately.

\begin{prop} \label{UniquenessWT}
  Let $\mathcal{L} \subset \N$ such that $\ds \sum_{n \in \mathcal{L}} \frac{1}{n} = \infty$. Assume that the following equalities hold
	$$
		M_q(\mu_n^2) = M_{\tilde{q} - C(\lambda^2+\frac{1}{4})}(\mu_n^2), \quad \forall n \in \mathcal{L}.
	$$
	Then both following equalities hold
	\begin{eqnarray*}
	  \Delta_q(\mu^2) = \Delta_{\tilde{q} - C(\lambda^2+\frac{1}{4})}(\mu^2), \quad \forall \mu \in \C, \\
		\delta_q(\mu^2) = \delta_{\tilde{q} - C(\lambda^2+\frac{1}{4})}(\mu^2), \quad \forall \mu \in \C.
	\end{eqnarray*}
\end{prop}

\begin{proof}
Under our assumption and using the definition of the generalized WT function, we thus get
$$
	\delta_q(\mu_n^2) \Delta_{\tilde{q} - C(\lambda^2+\frac{1}{4})}(\mu_n^2) = \delta_{\tilde{q} - C(\lambda^2+\frac{1}{4})}(\mu_n^2) \Delta_q(\mu_n^2), \quad \forall n \in \mathcal{L}.
$$
Since the product of functions $\mu \mapsto \delta_q(\mu^2) \Delta_{\tilde{q} - C(\lambda^2+\frac{1}{4})}(\mu^2), \  \delta_{\tilde{q} - C(\lambda^2+\frac{1}{4})}(\mu^2) \Delta_q(\mu^2)$ still belongs to the Cartwright class $\mathcal{C}$ by Corollary \ref{Bounded-iR}, we deduce first from Proposition \ref{UniquenessC} that
\begin{equation} \label{qq1}
	\delta_q(\mu^2) \Delta_{\tilde{q} - C(\lambda^2+\frac{1}{4})}(\mu^2) = \delta_{\tilde{q} - C(\lambda^2+\frac{1}{4})}(\mu^2) \Delta_q(\mu^2), \quad \forall \mu \in \C.
\end{equation}
	
Observe now that the unitary relations (\ref{Unitary-DeltaExt}) imply that the functions $\Delta_q(\mu^2)$ and $\delta_q(\mu^2)$ cannot vanish simultaneously. Hence we deduce from (\ref{qq1}) that the functions $\Delta_q(\mu^2)$ and $\Delta_{\tilde{q} - C(\lambda^2+\frac{1}{4})}(\mu^2)$ (resp. $\delta_q(\mu^2)$ and $\delta_{\tilde{q} - C(\lambda^2+\frac{1}{4})}(\mu^2)$) share the same zeros counted with multiplicities respectively denoted by $(\alpha_n^2)_n$ and $(\beta_n^2)_n$.

Notice also that the functions $\mu^2 \mapsto \Delta_q(\mu^2), \ \delta_q(\mu^2)$ are entire of order $\half$ by Corollary \ref{Bounded-iR}. Hence we can use the Hadamard factorization Theorem (see \cite{Lev}) to write these functions as
\begin{eqnarray} \label{qq2}
  \Delta_q(\mu^2) = G \prod_{n \in \N} \left( 1 - \frac{\mu^2}{\alpha_n^2} \right), \quad G = \Delta_q(0), \\
  \delta_q(\mu^2) = g \prod_{n \in \N} \left( 1 - \frac{\mu^2}{\beta_n^2} \right), \quad g = \delta_q(0). \nonumber
\end{eqnarray}

Recalling that the zeros of the functions $\Delta_q(\mu^2)$ and $\Delta_{\tilde{q} - C(\lambda^2+\frac{1}{4})}(\mu^2)$ (resp. $\delta_q(\mu^2)$ and $\delta_{\tilde{q} - C(\lambda^2+\frac{1}{4})}(\mu^2)$) coincide with the same multiplicities, we thus deduce from (\ref{qq1}) and (\ref{qq2})
\begin{eqnarray} 
  \frac{\Delta_q(\mu^2)}{\Delta_{\tilde{q} - C(\lambda^2+\frac{1}{4})}(\mu^2)} = \frac{G}{\tilde{G}}, \quad \forall \mu \in \C, \label{qq3} \\
  \frac{\delta_q(\mu^2)}{\delta_{\tilde{q} - C(\lambda^2+\frac{1}{4})}(\mu^2)} = \frac{g}{\tilde{g}}, \quad \forall \mu \in \C. \label{qq4}
\end{eqnarray}
	
We now use the asymptotics given in Proposition \ref{Asymp-WT} together with (\ref{qq3}) - (\ref{qq4}) and our main assumption to show successively that $A = \tilde{A}$, $G = \tilde{G}$ and $g = \tilde{g}$. We conclude that
\begin{eqnarray*}
	\Delta_q(\mu^2) = \Delta_{\tilde{q} - C(\lambda^2+\frac{1}{4})}(\mu^2), \quad \forall \mu \in \C, \\
	\delta_q(\mu^2) = \delta_{\tilde{q} - C(\lambda^2+\frac{1}{4})}(\mu^2), \quad \forall \mu \in \C.
\end{eqnarray*}

\end{proof}

\subsection{The inverse problem for the radial equation}

At this stage, we have proved from our main assumption that (see Proposition \ref{UniquenessWT})
\begin{eqnarray*}
	\Delta_q(\mu^2) = \Delta_{\tilde{q} - C(\lambda^2+\frac{1}{4})}(\mu^2), \quad \forall \mu \in \C, \\
  \delta_q(\mu^2) = \delta_{\tilde{q} - C(\lambda^2+\frac{1}{4})}(\mu^2), \quad \forall \mu \in \C.
\end{eqnarray*}
In particular, we have proved that the generalized WT functions
$$
	M_q(\mu^2) = M_{\tilde{q} - C(\lambda^2+\frac{1}{4})}(\mu^2), \quad \forall \mu \in \C \setminus \{ (\alpha_n)_{n \in \Z} \},
$$
coincide. We can thus use the celebrated Borg-Marchenko Theorem in the form given in \cite{FY} (see the details below) to prove that
$$
  A = \tilde{A}, \quad q(x) = \tilde{q}(x) - C(\lambda^2+\frac{1}{4}), \quad \forall x \in (0,A).
$$
Using the definition of $q$, we thus finally get
$$
  A = \tilde{A}, \quad \tilde{a}(x) = a(x) - C, \quad \forall x \in (0,A).
$$

In conclusion, under our main assumption, we have proved successively that there exists a constant $C$ such that
$$
  b(y)= \tilde{b}(y)+C  , \quad \forall y \in (0,B), \quad A = \tilde{A}, \quad  a(x) = \tilde{a}(x) + C, \quad \forall x \in (0,A).
$$
Recalling that the metric $g$ of the Liouville surface is
$$
  g = (a(x) - b(y)) [ dx^2 + dy^2],
$$
we see immediately that
$$
  \tilde{g} = g,
$$
which finishes the proof of Theorem \ref{Main}. \\

For completeness and since the Borg-Marchenko Theorem used above is not fully standard, we recall here its proof adapted from the one given by Freiling and Yurko in \cite{FY}.

\begin{thm} \label{BorgMarchenko}
  Let $M_q$ and $M_{\tilde{q}}$ be the two generalized WT functions associated to the equations
	\begin{equation} \label{Eq}
	  -u'' + q(x) u = -\mu^2 u, \quad \forall x \in (0,A), \quad -u'' + \tilde{q}(x) u = -\mu^2 u, \quad \forall x \in (0,\tilde{A}),
	\end{equation}
	where $q, \tilde{q}$ satisfy the asymptotics (\ref{Asymp}). If
	$$
	  M_q(\mu^2) = M_{\tilde{q}}(\mu^2), \quad \forall \mu \in \C \setminus \{ \textrm{poles}\},
	$$
	then $A = \tilde{A}$ and
	$$
	  q(x) = \tilde{q}(x), \quad \forall x \in (0,A).
	$$
\end{thm}

\begin{proof}
From the asymptotics given in Proposition \ref{Asymp-WT}, we easily get $A = \tilde{A}$. Also recall from the proof of Proposition \ref{UniquenessWT} that the equalities between the generalized WT functions imply the equalities between the characteristic functions, \textit{i.e.} we also have
\begin{equation} \label{rr1}
  \Delta_q(\mu^2) = \Delta_{\tilde{q}}(\mu^2).
\end{equation}

Let us define the matrix
$$
  P(x,\mu^2) = \left[ \begin{array}{cc} P_{11}(x,\mu^2) & P_{12}(x,\mu^2) \\
	                    P_{21}(x,\mu^2) & P_{22}(x,\mu^2) \end{array} \right],
$$
by the formula
$$
  P(x,\mu^2) \left[ \begin{array}{cc} \tilde{S}_{10}(x,\mu^2) & \tilde{\phi}(x,\mu^2) \\
	                    \tilde{S'}_{10}(x,\mu^2) & \tilde{\phi}'(x,\mu^2) \end{array} \right] = \left[ \begin{array}{cc} S_{10}(x,\mu^2) & \phi(x,\mu^2) \\
	                    S'_{10}(x,\mu^2) & \phi'(x,\mu^2) \end{array} \right],
$$
where
\begin{equation} \label{phi}
  \phi = -\frac{S_{11}}{\Delta_q(\mu^2)} = S_{20} + M_q(\mu^2) S_{10}.
\end{equation}
 Remark that $W(\tilde{S}_{10}, \tilde{\phi}) = 1$. Hence $\{\tilde{S}_{10}, \tilde{\phi}\}$ is a fundamental system for the radial  equation and the matrix
$$
  \left[ \begin{array}{cc} \tilde{S}_{10}(x,\mu^2) & \tilde{\phi}(x,\mu^2) \\
	                    \tilde{S'}_{10}(x,\mu^2) & \tilde{\phi}'(x,\mu^2) \end{array} \right],
$$
is invertible. We obtain
\begin{eqnarray}
  P_{11}(x,\mu^2) = S_{10}(x,\mu^2) \tilde{\phi}'(x,\mu^2) - \tilde{S}'_{10}(x,\mu^2) \phi(x,\mu^2), \label{rrr}  \\
	P_{12}(x,\mu^2) = \phi(x,\mu^2) \tilde{S}_{10}(x,\mu^2) - S_{10}(x,\mu^2) \tilde{\phi}(x,\mu^2). \label{rrrr}
\end{eqnarray}

On one hand, after some calculations and using the definition (\ref{phi}) of $\phi$, we have the following expressions for $P_{11}$ and $P_{12}$
\begin{eqnarray*}
  P_{11}(x,\mu^2) = & S_{10}(x,\mu^2) \tilde{S'}_{20}(x,\mu^2) - S_{20}(x,\mu^2) \tilde{S'}_{10}(x,\mu^2) \\
	                  & \quad + (M_{\tilde{q}}(\mu^2) - M_q(\mu^2)) S_{10}(x,\mu^2) \tilde{S'}_{10}(x,\mu^2),   \\
	P_{12}(x,\mu^2) = & S_{20}(x,\mu^2) \tilde{S}_{10}(x,\mu^2) - S_{10}(x,\mu^2) \tilde{S}_{20}(x,\mu^2) \\
	                  & \quad + (M_q(\mu^2) - M_{\tilde{q}}(\mu^2)) S_{10}(x,\mu^2) \tilde{S}_{10}(x,\mu^2).
\end{eqnarray*}
Using our main hypothesis $M_{\tilde{q}}(\mu^2) = M_q(\mu^2)$, we thus get
\begin{eqnarray}
  P_{11}(x,\mu^2) = S_{10}(x,\mu^2) \tilde{S'_{20}}(x,\mu^2) - S_{20}(x,\mu^2) \tilde{S'_{10}}(x,\mu^2), \label{rr2}  \\
	P_{12}(x,\mu^2) = S_{20}(x,\mu^2) \tilde{S_{10}}(x,\mu^2) - S_{10}(x,\mu^2) \tilde{S_{20}}(x,\mu^2). \label{rr3}
\end{eqnarray}

% ----------------   Modification de la preuve avec Pragmen Lindelof   -------------------

We conclude that the functions $P_{11}(x,\mu^2)$ and $P_{12}(x,\mu^2)$ are entire in the variable $\mu$. Moreover, using Corollaries \ref{UniEst-Sj0} and \ref{UniEst-Sj1}, we see that these functions are of exponential type and bounded on the imaginary axis.

On the other hand, for real $\mu$,  using (\ref{rr1}), (\ref{phi}) and (\ref{rrr}), (\ref{rrrr}), we can express alternatively the functions $P_{11}(x,\mu^2)$ and $P_{12}(x,\mu^2)$ as
\begin{eqnarray}
  P_{11}(x,\mu^2) & = & 1 + \frac{S_{11}(x,\mu^2)}{\Delta_q(\mu^2)} ( \tilde{S'}_{10}(x,\mu^2) - S'_{10}(x,\mu^2)) \nonumber \\
	                &   &  \quad \quad + \frac{S_{10}(x,\mu^2)}{\Delta_q(\mu^2)} ( \tilde{S'}_{11}(x,\mu^2) - S'_{11}(x,\mu^2)), \label{rr4} \\
	P_{12}(x,\mu^2) & = & \frac{1}{\Delta_q(\mu^2)} ( S_{10}(x,\mu^2) \tilde{S}_{11}(x,\mu^2) - S_{11}(x,\mu^2) \tilde{S}_{10}(x,\mu^2)). \label{rr5}
\end{eqnarray}
From the estimates in Corollaries \ref{UniEst-Sj0} and \ref{UniEst-Sj1} on the $S_{jm}(x,\mu^2)$ and $\tilde{S}_{jm}(x,\mu^2)$ for $j=1,2$ and $m=0,1$ and from the estimate (\ref{Delta-Minoration}), we can prove from (\ref{rr4}) - (\ref{rr5}) that for a fixed $x \in (0,A)$ and as $\mu \to \pm \infty$
\begin{equation} \label{rr6}
  |P_{11}(x,\mu^2) - 1| = O\left( \frac{1}{\log(|\mu|)^{\epsilon}} \right), \quad |P_{12}(x,\mu^2)| = 0\left( \frac{1}{\log(|\mu|)^{\epsilon}} \right),
\end{equation}
where $\epsilon = min (\eo, \eu)$. Applying the Phragmen-Lindel\"of's theorem (\cite{Boa}, Thm 1.4.2) on each quadrant of the complex plane, we deduce that 
$P_{11}(x,\mu^2)$ and $P_{12}(x,\mu^2)$ are bounded on $\C$. We thus deduce from the Liouville Theorem and (\ref{rr6}) that
\begin{equation} \label{rr7}
  P_{11}(x,\mu^2) = 1, \quad P_{12}(x,\mu^2) = 0, \quad \forall \mu \in \C.
\end{equation}

But recall from the definition of $P(x,\mu^2)$ that
\begin{eqnarray*}
  S_{10}(x,\mu^2) = P_{11}(x,\mu^2) \tilde{S}_{10}(x,\mu^2) + P_{12}(x,\mu^2) \tilde{S'}_{10}(x,\mu^2), \\
	\phi(x,\mu^2) = P_{11}(x,\mu^2) \tilde{\phi}(x,\mu^2) + P_{12}(x,\mu^2) \tilde{\phi}'(x,\mu^2).
\end{eqnarray*}
Hence we deduce from this and (\ref{rr7}) that
$$
  S_{10}(x,\mu^2) = \tilde{S}_{10}(x,\mu^2), \quad \phi (x,\mu^2) = \tilde{\phi}(x,\mu^2).
$$
Since $\{S_{10}, \phi\}$ and $\{\tilde{S}_{10}, \tilde{\phi} \}$ are FSS of the equations (\ref{Eq}), we conclude that $q(x) = \tilde{q}(x), \quad \forall x \in (0,A)$.

\end{proof}

\Section{An open problem} \label{Open}

We finish this paper with an interesting open problem concerning the inverse scattering problem for asymptotically hyperbolic Liouville surfaces. We have shown in our main Theorem \ref{Main} that the knowledge of the reflection operators $L(\lambda)$ or $R(\lambda)$ at a fixed non-zero energy $\lambda$ is enough to determine the metric of asymptotically hyperbolic Liouville surfaces. It would be interesting to address the question whether the transmission operators $T_L(\lambda)$ or $T_R(\lambda)$ at a fixed energy can be also used to determine uniquely the metric. According to Proposition \ref{WT-Reflection}, this is equivalent to know whether the characteristic operator $\Delta(\lambda)$ allows one to recover the metric.

Starting from the assumption
\begin{equation} \label{As-Delta}
  \Delta(\lambda) = \tilde{\Delta}(\lambda), \quad \lambda \not=0, \ \ \textrm{fixed},
\end{equation}
we can follow the arguments in the first subsections of Section \ref{IP} to deduce that there exists a constant $C \in \R$ such that
$$
  b(y) = \tilde{b}(y) + C, \quad \forall y \in (0,B),
$$
$$
  \mu_n^2 = \tilde{\mu}_n^2 + C (\lambda^2 + \frac{1}{4}), \quad \forall n \in \N,
$$
and
\begin{equation} \label{As-Delta1}
  \forall n \in \N, \quad \Delta_q(\mu_n^2) = \Delta_{\tilde{q}}(\mu_n^2 - C (\lambda^2 + \frac{1}{4})).
\end{equation}
Moreover, using Lemma \ref{Sym}, (\ref{As-Delta1}) leads to
\begin{equation} \label{As-Delta2}
  \forall n \in \N, \quad \Delta_q(\mu_n^2) = \Delta_{\tilde{q} - C (\lambda^2 + \frac{1}{4})}(\mu_n^2).
\end{equation}
Finally, using the fact that the functions $\mu \mapsto \Delta_q(\mu_n^2), \Delta_{\tilde{q} - \lambda^2C}(\mu_n^2)$ are in the Cartwright class $\mathcal{C}$, we deduce from (\ref{MuntzCond}), Proposition \ref{UniquenessC} and (\ref{Assump1}) that
\begin{equation} \label{Assump1}
  \forall \mu \in \C, \quad \Delta_q(\mu^2) = \Delta_{\tilde{q} - C (\lambda^2 + \frac{1}{4})}(\mu^2).
\end{equation}

We can first use the asymptotics in Proposition \ref{Asymp-WT} to prove that $A = \tilde{A}$ from (\ref{Assump1}).

The remaining problem is thus to show that the potentials $q(x)$ and $\tilde{q} - C (\lambda^2 + \frac{1}{4})$ coincide on $(0,A)$ from (\ref{Assump1}). Since the characteristic function $\mu^2 \mapsto \Delta_q(\mu^2)$ is entire of order $\half$, we have by the Hadamard Theorem (see (\ref{qq2}))
$$
  \Delta_q(\mu^2) = G \prod_{n \in \N} \left( 1 - \frac{\mu^2}{\alpha_n^2} \right), \quad G = \Delta_q(0).
$$
In other words, the characteristic function $\Delta_q(\mu^2)$ is completely determined (up to the constant $G$) by the eigenvalues of the radial separated ODE with the boundary conditions (\ref{BC}). We deduce then from (\ref{Assump1}) that the potentials $\tilde{q} - C (\lambda^2 + \frac{1}{4})$ s  in the same isospectral class than $q$.

%has the same zeros $(\alpha_n)_{n \in \Z}$ than the potential $q$.

Note that, for selfadjoint (possibly singular) Schr\"odinger operators with spectral gap in their spectrum, we can perform explicit transformation of the initial potential (Crum-Darboux or B\"acklund-Darboux like transformations) that lead to new potentials that preserve the spectrum of the associated Schrodinger operator. We refer to \cite{GST} and the references therein for a detailed constructions of these isospectral potentials.

By analogy, we conjecture that there might exist many isospectral potentials satisfying (\ref{Assump1}) for our non selfadjoint radial equation and that, among these isospectral potentials, there might be some satisfying the hypothesis in Definition \ref{AHLS}, leading to non isometric asymptotically hyperbolic Liouville surfaces having the same transmission operator at a fixed energy. That is we conjecture a non-uniqueness result in the inverse scattering problem from the transmission operators at a fixed energy.

\end{document}